\documentclass{article}
\usepackage[utf8]{inputenc}
\usepackage[T1]{fontenc}
\usepackage{amsmath,amssymb,amsthm,mathtools}
\usepackage[letterpaper,margin=1in]{geometry}
\usepackage{graphicx,enumitem,booktabs,microtype}
\usepackage{tikz}
\usepackage{float}
\usepackage[dvipsnames]{xcolor}
\usepackage[ruled,noend,linesnumbered]{algorithm2e}
\usepackage[colorlinks=true,linkcolor=MidnightBlue,citecolor=MidnightBlue,urlcolor=MidnightBlue]{hyperref}
\hypersetup{
  pdftitle={On the Strong Matroid Secretary Conjecture and Beyond},
  pdfauthor={Hamed Abdi, Kiarash Banihashem, MohammadTaghi Hajiaghayi, Danny Mittal},
  pdfsubject={Ordinal matroid secretary and single-sample prophet inequalities},
  pdfkeywords={matroid secretary, prophet inequalities, online algorithms, ordinal information}
}
\setlist{itemsep=2pt,topsep=4pt,partopsep=0pt,parsep=0pt}
\let\originalthebibliography\thebibliography
\renewcommand{\thebibliography}[1]{%
  \originalthebibliography{#1}%
  \addcontentsline{toc}{section}{\refname}%
}
\usepackage{color-edits}
\addauthor{kb}{cyan}
\newtheorem{theorem}{Theorem}[section]
\newtheorem{lemma}[theorem]{Lemma}
\newtheorem{proposition}[theorem]{Proposition}
\newtheorem{corollary}[theorem]{Corollary}
\theoremstyle{definition}
\newtheorem{definition}[theorem]{Definition}
\newtheorem{conjecture}[theorem]{Conjecture}
\theoremstyle{remark}

\newcommand{\E}{\mathbb E}
\newcommand{\F}{\mathbb F}

\newcommand{\ind}{\mathbf 1}
\newcommand{\OPT}{\operatorname{OPT}}
\newcommand{\cl}{\operatorname{cl}}
\newcommand{\spn}{\operatorname{span}}
\newcommand{\rank}{\operatorname{rank}}
\newcommand{\dist}{\mathrel{\overset{\mathrm d}=}}

\title{On the Strong Matroid Secretary Conjecture and Beyond%
\thanks{This is a preliminary version. Due to time constraints it is
being posted earlier than we would have liked, in order to record the
results; the writing will be improved in future versions.}}
\author{\begin{tabular}{cc}
Hamed Abdi & Kiarash Banihashem\\[0.5ex]
MohammadTaghi Hajiaghayi & Danny Mittal\\[1ex]
\multicolumn{2}{c}{\small University of Maryland, College Park}
\end{tabular}}
\date{September 16, 2026}
\begin{document}
\maketitle
\begin{abstract}
The strong matroid secretary conjecture asserts that every matroid admits
a $1/e$-competitive secretary algorithm, matching the classical
single-choice guarantee. We formulate a finite linear program whose value
is the optimal ordinal competitive ratio of any fixed matroid; for all
matroids of positive rank on seven elements and nearly all on eight, this
value exceeds $1/e$. The same computations suggested that the optimal
ratio is monotone under truncation of the matroid; we prove this for
uniform matroids, where the ratio is strictly increasing in the rank, and
refute it for a graphic matroid. Guided by this
evidence, we prove the conjecture for every linear matroid, a class that
includes graphic matroids, regular matroids, laminar matroids, and gammoids, giving a
$1/e$-competitive ordinal secretary
algorithm. The algorithm maintains bounds on the expected intersection
dimension of the accepted span with every ambient subspace. Uncrossing and separation show
that these bounds can be preserved while admitting each current greedy-basis
element with a prescribed probability and the construction uses finite linear
programs.

For every matroid, we also give a single-sample prophet algorithm with
competitive ratio $1/2$ in any fixed arrival order independent of the samples
and values. Its output, including the selected values, has exactly the law of
an independent fair thinning of an optimum from a fresh product draw. The
algorithm uses $O(n^2)$ independence queries on $n$ elements. Both constants are tight in their
respective models. We also give a self-contained black-box reduction that
converts a single-sample prophet ratio $\alpha$ into a secretary ratio
$\alpha^2/16$, preserving polynomial running time. Our single-sample
algorithm consequently yields a $1/64$-competitive ordinal secretary
algorithm for arbitrary matroids.
\end{abstract}

\section{Introduction}

In the matroid secretary problem, an adversary assigns nonnegative weights to
the $n$ elements of a matroid. The elements then arrive in a uniformly random
order, and the algorithm learns an element's weight only when it arrives.
Each element must be accepted or rejected immediately and irrevocably, and
the accepted set must remain independent. The goal is to obtain, in
expectation, a constant fraction of the weight of a maximum-weight independent
set; the best such fraction guaranteed by an algorithm is its
\emph{competitive ratio}.

In rank one, where any single element is independent, this is the classical
secretary problem. The classical cutoff rule rejects roughly the first $n/e$
arrivals and then accepts the first \emph{record}, that is, the first arrival
heavier than everything seen before. It selects the heaviest element with
probability approaching $1/e$~\cite{ferguson1989}, and in the limit no
algorithm does better, even one that observes the numerical
weights~\cite[Theorem~15]{graphic2025}. Babaioff, Immorlica,
and Kleinberg~\cite{babaioff2007,babaioff2018} introduced the matroid
generalization and asked whether a constant competitive ratio is possible for
every matroid. The \emph{strong matroid secretary conjecture} asks for the
same constant as in rank one.

\begin{conjecture}[Strong matroid secretary conjecture]\label{conj:strong}
Every matroid admits a $1/e$-competitive secretary algorithm.
\end{conjecture}

Two clarifications are needed. First, since $1/e$ is optimal in rank one,
the conjecture asks for the best constant that could hold for every matroid.
Second, an algorithm is \emph{ordinal} if it uses only the relative order of
the weights observed so far, together with the matroid; it never uses
numerical values. In some matroids of higher rank, numerical weights allow a
better ratio than ordinal information does~\cite{chan2015}. Throughout, we
read the conjecture as a statement about ordinal algorithms, which is the
stronger reading, and all our secretary algorithms are ordinal.

Before this work, the conjecture was open even without the constant: the
best general guarantee was of order $1/\log\log r$, where $r$ is the
rank~\cite{lachish2014,feldman2015}. Corollary~\ref{cor:general-secretary}
below gives a constant ratio for every matroid, and Theorem~\ref{thm:linear}
gives the conjectured $1/e$ for every linear matroid. Previously, the
strong conjecture was known for uniform matroids: the ordinal virtual
algorithm of Babaioff, Immorlica, Kempe, and
Kleinberg~\cite{babaioff2007knapsack} is $1/e$-competitive for every rank,
and Corollary~\ref{cor:uniform-strong} in Appendix~\ref{app:truncation}
recovers this as a consequence of strict monotonicity in the rank. It was
also known for transversal matroids: the online bipartite-matching
algorithm of Kesselheim, Radke, T\"onnis, and
V\"ocking~\cite{kesselheim2013} is asymptotically $1/e$-competitive and
admits an ordinal implementation in this setting, and Soto, Turkieltaub, and
Verdugo~\cite{soto2021} give the exact $1/e$ per-element guarantee. For
graphic matroids, the previously best known ratio was $1/3.95$, approaching
$1/e$ only as the girth grows~\cite{graphic2025}.

To see where the difficulty lies, recall why the cutoff rule works. Suppose
the heaviest element arrives at time $t$ after the cutoff $s$. It is a record,
and it is accepted exactly when no record appeared at times $s+1,\ldots,t-1$,
which happens when the heaviest of the first $t-1$ arrivals lies among the
first $s$; this has probability $s/(t-1)$. Nothing else can go wrong: if no
record appeared at those times, nothing has been accepted, so the slot is
free when the heaviest element arrives.
In a matroid, the natural candidate at time $t$ is an element of the greedy
basis of the observed set, and every element of the full greedy basis $B(w)$
is such a candidate whenever it has arrived. But the algorithm may already have accepted a set
that spans the candidate, even a small one. A bound on the number of accepted
elements does not prevent this. What must be controlled is the \emph{span} of
the accepted set, and since the adversary chooses the weights, it must be
controlled in every direction at once.

Our starting point was computational. For a fixed matroid, the optimal
ordinal competitive ratio is the value of a finite linear program, which we
describe in Section~\ref{sec:lp} and prove correct in Appendix~\ref{app:lp}.
To our knowledge
this is the first program that computes this ratio for an arbitrary fixed
matroid. We solved it for all matroids of positive rank on seven elements and
nearly all on eight, and in every case the optimal ratio exceeded
$1/e$.\footnote{Code and data for these computations are available at
\url{https://github.com/TlatoaniHJ/MatroidSecretary}; the repository is
released as is, for full disclosure.} This
evidence for the strong conjecture led us to seek a proof, and the structure of the optimal solutions guided the construction
we present.

We prove the strong conjecture for every linear matroid, using a
computationally unrestricted ordinal algorithm. A matroid is \emph{linear}
if it can be represented by vectors over some field, with independence
meaning linear independence. Linear matroids include essentially all of
the classes of matroids that are usually considered: uniform and partition
matroids, graphic
and cographic matroids, regular matroids (those representable over every
field), transversal matroids and more generally
gammoids~\cite{lindstrom1973,ingletonpiff1973}, and laminar matroids, which
are gammoids~\cite{finkelstein2011,fifeoxley2017}. For all of these,
Theorem~\ref{thm:linear} gives the ratio $1/e$, the best constant that can
hold for every matroid; previously, $1/e$ was known for uniform and
transversal matroids~\cite{babaioff2007knapsack,kesselheim2013,soto2021},
constant ratios were known for some other classes~\cite{babaioff2007,soto2021},
and no constant at all was known for linear matroids in general. Non-linear matroids exist, and in the enumerative sense almost all
matroids are non-linear~\cite{nelson2018}; for them,
Corollary~\ref{cor:general-secretary} still provides a constant ratio. The algorithm behind Theorem~\ref{thm:linear} controls exactly the
quantity identified above: the expected dimension of the intersection of the
accepted span with every subspace.

Our second result concerns the single-sample matroid prophet inequality.
In this model, each element has an independent random value drawn from a
distribution that is not revealed to the algorithm; instead, the algorithm
receives one independent sample from each distribution before the arrivals
begin. The arrival order can be arbitrary, provided it is fixed independently
of the samples and actual values. With full knowledge of the distributions,
the optimal guarantee for matroids is $1/2$~\cite{kleinberg2012}. We show
that one sample per distribution suffices for the same guarantee in every
matroid, through an algorithm whose output has an exact distributional
description: the selected labels and their values are distributed as an
independent fair thinning of an optimal basis from a fresh draw. A
self-contained reduction then converts this into a polynomial-time
$1/64$-competitive ordinal secretary algorithm for every matroid, settling
the original conjecture of Babaioff et al.\ that every matroid admits a
constant competitive ratio.\footnote{Independently and concurrently,
Singla~\cite{singla2026} also resolved the constant-ratio conjecture, with
an ordinal algorithm that accepts each element of the optimum with
probability at least $1/4$ for every matroid, without knowing the matroid
in advance. The two works are incomparable: for linear matroids we obtain
the optimal constant $1/e$, while for general matroids our constant $1/64$
is weaker than his $1/4$. Singla also obtains a constant-competitive
single-sample matroid prophet inequality, with ratio $1/8$;
Theorem~\ref{thm:prophet} gives the tight ratio $1/2$.}

\subsection{Our results}

Throughout, the algorithm is given the ground set and matroid before
arrivals. Let $B(w)$ be the greedy basis under a fixed rule for breaking
ties; its weight is the optimum $\OPT(w)$. Table~\ref{tab:results}
summarizes the results.

\begin{table}[H]
\centering
\begin{tabular}{llll}
\toprule
Model & Matroids & Ratio & Computation\\
\midrule
Ordinal secretary & linear & $1/e$ (tight) & finite; no polynomial bound\\
Single-sample prophet & all & $1/2$ (tight) & $O(n^2)$ independence queries\\
Ordinal secretary & all & $1/64$ & polynomial\\
\bottomrule
\end{tabular}
\caption{Summary of results. The two constants are tight in their models:
$1/e$ for secretary algorithms and $1/2$ for single-sample prophet
inequalities, both already in rank one.}
\label{tab:results}
\end{table}

\begin{theorem}[Linear matroid secretary]\label{thm:linear}
Let $M$ be a linear matroid on $n\ge2$ elements.
For each $1\le s<n$, there is an ordinal secretary algorithm
(Algorithm~\ref{alg:linear}) returning
an independent set $A$ such that, for every nonnegative weight vector
$w$, each element of $B(w)$ is selected with probability exactly
\[
 c_n(s)=\frac{s}{n}\sum_{k=s}^{n-1}\frac1k.
\]
In particular, choosing $s$ from $n$ alone gives expected weight at least
$\OPT(w)/e$. The algorithm is a finite construction using linear programs;
no polynomial running-time bound is asserted.
\end{theorem}

The quantity $c_n(s)$ is exactly the probability with which the classical
cutoff rule with cutoff $s$ selects the heaviest of $n$ elements: the
heaviest element arrives at a uniform time $t$, and is selected when $t>s$
and the heaviest of the previous $t-1$ arrivals lies among the first $s$.
Thus $c_n(s)=\frac1n\sum_{t=s+1}^n\frac{s}{t-1}$, which tends to $1/e$ when
$s/n\to1/e$. Our algorithm gives \emph{every} element of the optimal basis
this same selection probability.

The construction works with a representation over a finite field. This
does not restrict the theorem's scope: every finite linear matroid has
such a representation by Rado's theorem~\cite{rado1957}, and one can find
it by exhaustive search from the matroid before arrivals begin.

The per-element conclusion is a \emph{probability-competitive} guarantee
in the terminology of Soto, Turkieltaub, and Verdugo~\cite{soto2021}. It is
stronger than a bound on the expected weight. If $H$ is any prefix of the
weight order, greedy selects exactly $\rank_M(H)$ elements from $H$, so our
policy satisfies $\E|A\cap H|\ge c_n(s)\rank_M(H)$ for every such prefix
simultaneously, without knowing which prefix will be used to evaluate it.
Summation by parts (Section~\ref{sec:preliminaries}) turns these prefix
bounds into the weighted bound. The asymptotic constant $1/e$ cannot be
improved, already in rank one, even by algorithms that observe numerical
weights~\cite[Theorem~15]{graphic2025}.

\begin{theorem}[Single-sample matroid prophet]\label{thm:prophet}
Let $M$ be any matroid, and let $(X_e)_{e\in E}$ be independent
nonnegative values. Given one independent sample from each value
distribution, an online algorithm selects an independent set $A$ satisfying
\[
 \E\sum_{e\in A}X_e=\frac12\E\OPT(X).
\]
The guarantee holds in every fixed order independent of the samples and
values. The algorithm uses at most $n+1$ greedy scans, and hence
$O(n^2)$ independence queries. No factor larger than $1/2$ holds uniformly,
even in rank one when the algorithm is given the distributions.
\end{theorem}

The equality is not a coincidence of the analysis. Section~\ref{sec:prophet}
proves that the selected labels together with their values have exactly the
law of an independent fair thinning of the optimal basis of a fresh draw from
the value distributions, so the factor $1/2$ is the thinning probability.

The two results have different computational requirements. The secretary
construction maintains constraints for a finite but potentially very large
collection of subspaces. The prophet algorithm updates one stored value
vector and recomputes a greedy basis when an element arrives. It uses only
an independence oracle for the matroid.

The single-sample result also yields a secretary guarantee for arbitrary
matroids.

\begin{corollary}[General matroid secretary]\label{cor:general-secretary}
Every matroid admits a polynomial-time $1/64$-competitive ordinal
secretary algorithm.
\end{corollary}

The implication from constant-sample prophet inequalities to a constant
secretary guarantee is recorded by Fu et al.~\cite{fu2024}, who attribute
it to Wenzheng Li (2023, personal communication). We did not have access to
that argument, so Theorem~\ref{thm:reduction} gives a self-contained
reduction converting any single-sample prophet ratio $\alpha$ into a
secretary ratio $\alpha^2/16$, while preserving polynomial running time.
Instantiating this reduction with our $1/2$-competitive algorithm proves
the corollary. The reduction uses only the expected prophet guarantee, not
the distributional identity.

\emph{Further results.}
Section~\ref{sec:lp} presents the linear program that computes the optimal
ordinal secretary ratio $\rho(M)$ of any fixed matroid $M$, the
computations we performed with it, and the conjecture they suggested;
Appendix~\ref{app:lp} proves that its value is exactly $\rho(M)$
(Theorem~\ref{thm:lp-exact}), so in particular the supremum defining
$\rho(M)$ is attained. Appendix~\ref{app:truncation} studies how $\rho$
behaves under truncation of the matroid. For uniform matroids the ratio is strictly
increasing in the rank (Theorem~\ref{thm:uniform-monotonicity}), which
recovers the known $1/e$ guarantee for uniform
matroids~\cite{babaioff2007knapsack} (Corollary~\ref{cor:uniform-strong});
a graphic matroid shows that this monotonicity fails in general
(Theorem~\ref{thm:truncation-counterexample}).

\subsection{Techniques}\label{sec:techniques}

\emph{The invariant for linear matroids.}
Fix a cutoff $s$ and a linear matroid represented by vectors $v_e$. Write
$W(A)$ for the span of the vectors of an accepted set $A$. Our policy will
accept each greedy-basis candidate arriving at time $t>s$ with probability
exactly $a_t=s/(t-1)$, conditional on the observed set and on the identity of
the arrival, which gives every element of the optimum the classical
probability $c_n(s)$. For this, the candidate must be available: its vector
must not lie in $W(A)$. In rank one, the relevant statistic is the
probability that the slot is used, and the cutoff rule keeps it equal to
$\beta_t=1-s/t$ after $t$ arrivals. For a linear matroid we maintain, for
\emph{every} subspace $U$ of the ambient space,
\[
 \E\dim(W(A)\cap U)\le\beta_t\dim U .
\]
Applied to the line spanned by a candidate $v_e$, this says that the
candidate is blocked with probability at most $\beta_{t-1}$, so at least
$1-\beta_{t-1}=a_t$ of the probability mass is free to accept it. This is
exactly the mass we need to spend. Constraints on lines of ground vectors
alone would not suffice; Section~\ref{sec:linear-invariant} gives an example,
and the proof genuinely uses constraints at subspaces that no set of ground
vectors generates.

\emph{Preserving the invariant.}
Spending mass $a_t$ on accepting $v_e$ raises the load $\dim(W(A)\cap U)$ by
one at each accepting state $A$ with $v_e\in W(A)+U$. The question is whether
the mass can be placed so that no constraint is violated. Along a
\emph{chain} of subspaces the sets of expensive states are nested, so there
is a single order of the states that is cheapest for every member of the
chain at once; placing the mass in that order preserves all the chain's
constraints, by a short calculation that uses the invariant at $U$ and at
$U+\langle v_e\rangle$. An arbitrary family of subspaces admits no common
order. We therefore consider a nonnegative combination of the constraints and
\emph{uncross} it: replacing two incomparable subspaces $U,T$ by $U\cap T$
and $U+T$ keeps the total dimension and can only increase the total load, so
every combination is dominated by a chain. Hence every single combination can
be satisfied, and a minimax argument (separation in a finite-dimensional
space) yields one placement satisfying all constraints simultaneously. This
last step works only \emph{on average} over the identity of the last
arrival, because the candidates form an independent set and an independent
set has at most $\dim U$ vectors in $U$. That average is exactly what a
uniformly random arrival order provides.

\emph{From the recursion to an online policy.}
The argument defines, for each observed set $S$ with its relative weight
order, a law $\mu_S$ on accepted sets, together with transition rules that
modify each $\mu_{S-e}$ by admitting $e$; the average of the modified laws
over $e$ is $\mu_S$. The online policy, when $e$ arrives with
observed set $S$ and accepted set $A$, flips a coin whose bias is read off
the transition rule at $A$. Because the recursion uses only the relative
weight order on $S$, everything it needs is available at decision time.
Section~\ref{sec:linear} presents the policy as pseudocode first and then
proves that it is well defined and has the claimed guarantee; in rank one it
reduces exactly to the classical cutoff rule.

\emph{Maintaining a thinned optimum.}
A \emph{fair thinning} of a set retains each of its elements independently
with probability $1/2$. The prophet algorithm stores a vector, initially the
samples, and keeps its accepted set inside the greedy basis of the stored
vector. When an element arrives, the algorithm may overwrite its stored
coordinate with the actual value, and it couples this overwrite with the
acceptance decision so that two properties hold at all times: the stored
vector has the original product distribution, and conditional on it the
accepted set is a fair thinning of the processed part of its greedy basis.
A one-coordinate change moves a greedy basis by at most one exchange, so only
the arriving element and a possibly displaced basis element are involved. The
critical case is a displaced element that has already arrived: if it was
accepted, the new element must be rejected, and if it was rejected, the new
element is accepted with probability one; since the earlier acceptance was a
fair bit, the average is $1/2$. The proof fixes the two draws of each
coordinate as an unordered pair and tracks the fair bit that says which one is
currently stored.

\emph{Secretary guarantees from samples.}
To use a prophet algorithm for the secretary problem, a rejected prefix of
the arrivals plays the role of the samples. The obstacle is that the prophet
algorithm may later decide to select one of those already rejected elements,
which cannot be implemented. Three moves handle this. A preliminary
half-sample filters the instance to elements not spanned by heavier sample
elements, so that the expected total weight of all remaining candidates is at
most the optimum, while half the optimum survives. Then each remaining
element is independently \emph{activated} with a small probability $q$, both
as a sample and as a value, so that the useful prophet reward is linear in
$q$ while the unimplementable selections, which require both activations,
cost only $q^2$. Finally, an exact interleaving of the stored prefix with the
physical arrivals produces a virtual arrival order independent of the
activations, as the prophet guarantee requires. Optimizing $q$ gives
$\alpha^2/16$.

\subsection{Related work}\label{sec:related}

Babaioff et al.~\cite{babaioff2007,babaioff2018} introduced the matroid
secretary problem and obtained constant ratios for several classes. For
general matroids, Lachish~\cite{lachish2014} and Feldman, Svensson, and
Zenklusen~\cite{feldman2015} gave algorithms with competitive ratio
$\Omega(1/\log\log r)$, the best known before this work. Concurrently and
independently, Singla~\cite{singla2026} obtained a $1/4$
probability-competitive ordinal algorithm for all matroids and a
$1/8$-competitive single-sample prophet inequality; see the footnote in the
introduction for a comparison. For uniform matroids,
Babaioff, Immorlica, Kempe, and Kleinberg~\cite{babaioff2007knapsack} gave
ordinal $1/e$-competitive algorithms for every rank,
Kleinberg~\cite{kleinberg2005} obtained a ratio of $1-O(r^{-1/2})$, and Chan,
Chen, and Jiang~\cite{chan2015} determined optimal thresholds for small
capacities, showing in particular that numerical information improves the
optimal ratio for rank two. For graphic matroids, Banihashem et
al.~\cite{graphic2025} obtained a $1/3.95$ guarantee and a guarantee
approaching $1/e$ as the girth grows, together with a matching high-girth
hardness bound; we use that hardness result in
Appendix~\ref{app:truncation}. Soto, Turkieltaub, and
Verdugo~\cite{soto2021} introduced probability-competitive guarantees and a
forbidden-set framework giving them for graphic, laminar, and other classes,
including the optimal $1/e$ for transversal matroids, where Kesselheim,
Radke, T\"onnis, and V\"ocking~\cite{kesselheim2013} had obtained an
asymptotic $1/e$ for the expected weight.
Buchbinder, Jain, and Singh~\cite{buchbinder2014} characterized optimal
policies for the classical and multiple-choice secretary problems by linear
programs; Section~\ref{sec:lp} and Appendix~\ref{app:lp} extend the
approach to arbitrary matroids.

Prophet inequalities compare against the expected optimum of independent
values with known distributions. Kleinberg and Weinberg~\cite{kleinberg2012}
proved the tight $1/2$ matroid prophet inequality. Azar, Kleinberg, and
Weinberg~\cite{azar2014} introduced prophet inequalities with limited sample
access, and Rubinstein, Wang, and Weinberg~\cite{rubinstein2020} showed that
one sample suffices for the tight rank-one guarantee. For general matroids,
Fu et al.~\cite{fu2024} obtained a $(1/4-\varepsilon)$ guarantee from
$O_\varepsilon(\log^4 n)$ samples per distribution, and Feldman, Svensson,
and Zenklusen~\cite{feldman2025samples} reduced the requirement to
$O_\varepsilon(\log n+\log r\,\log^2\log r)$ samples through sample-based
online contention resolution. Those guarantees permit arrival orders that
depend on the realized samples and values, whereas ours requires the order to
be fixed independently of them; in exchange we use one sample and obtain the
optimal constant.

\emph{Organization.}
Section~\ref{sec:preliminaries} fixes notation and the two models.
Section~\ref{sec:lp} describes the linear program for the optimal ordinal
ratio, the computations, and the truncation conjecture.
Section~\ref{sec:linear} presents the linear-matroid secretary algorithm and
proves Theorem~\ref{thm:linear}; Section~\ref{sec:prophet} does the same
for the prophet algorithm and Theorem~\ref{thm:prophet}; and
Section~\ref{sec:reduction} derives the secretary consequences.
Appendix~\ref{app:lp} states the program in full and proves its
correctness. Appendix~\ref{app:truncation} proves that the optimal ratio
is monotone under truncation for uniform matroids and that this fails for
a graphic matroid.

\section{Preliminaries}\label{sec:preliminaries}

Let $M=(E,\mathcal I)$ be a finite matroid with rank function $r$ and
closure operator $\cl_M$. We allow loops and parallel elements. The ground
set and matroid are known in advance; algorithms may query independence
of arbitrary subsets of $E$. We write $S+e$ and $S-e$ for adjoining and
removing an element. For a nonnegative vector $w$, put
$w(S)=\sum_{e\in S}w_e$ and
$\OPT(w)=\max_{I\in\mathcal I}w(I)$; $\OPT(w|_S)$ denotes the optimum
on the restriction $M|S$. All logarithms are natural.

Fix a total priority order on labels, independently of all weights and
randomness. Greedy scans elements in decreasing weight, using this priority
to break ties, and accepts an element exactly when it preserves
independence. We scan zero-weight elements as well, so the output $B(w)$
is a basis. For a restriction to $S$, write $B_w(S)$, or simply $B(S)$
when the weights are understood.

Three standard greedy properties will be used repeatedly. First,
$w(B(S))=\OPT(w|_S)$. Second, if $e\in B(S)$ and $e\in T\subseteq S$,
then $e\in B(T)$. Indeed, greedy selects $e$ precisely when the
elements preceding it in the greedy order do not span it; deleting elements
cannot destroy this property. Third, for every prefix $H$ of the greedy
order on $S$, $|B(S)\cap H|=r(H)$.

In the secretary model the weights are fixed before an independent uniform
permutation of $E$ is drawn. An arrival is a \emph{record} if it precedes every earlier arrival in
the weight order, including the fixed tie priority. A policy returning $A$ is $c$-competitive if
$\E w(A)\ge c\,\OPT(w)$ for every nonnegative weight vector $w$.
An \emph{ordinal} policy uses only the strict relative
order of weights revealed so far, with the fixed tie rule. A guarantee
$\E|A\cap H|\ge c r(H)$ for every weight prefix implies a weighted
$c$-competitive guarantee. To see this, list the elements in weight order
as $e_1,\ldots,e_n$, put $H_j=\{e_1,\ldots,e_j\}$, and set
$w_{e_{n+1}}=0$. Summation by parts gives
\begin{equation}\label{eq:layer-cake}
 \E w(A)=\sum_{j=1}^n(w_{e_j}-w_{e_{j+1}})\E|A\cap H_j|,
 \qquad
 \OPT(w)=\sum_{j=1}^n(w_{e_j}-w_{e_{j+1}})r(H_j).
\end{equation}
For ordinal policies the converse holds as well: approach the indicator
of a prefix by positive weight vectors preserving the given strict order.
The decision law stays fixed, so taking the limit recovers that prefix's
inequality.

In the single-sample prophet model, $X_e$ has distribution $\mathcal D_e$,
independently across $e$. The algorithm first receives independent
samples $S_e\sim\mathcal D_e$, independent also of $X$. It does not know
the distributions. The arrival order is fixed independently of all samples,
values, and internal random coins; the algorithm need not know it in advance.
A random order independent of these
variables is also permitted, by conditioning on the order. Competitive
ratios compare $\E\sum_{e\in A}X_e$ with $\E\OPT(X)$; one may assume
the latter is finite. The distributional identities below remain valid
without this integrability assumption.

\section{The optimal ordinal ratio as a linear program}\label{sec:lp}

For a matroid $M$ of positive rank, let $\rho(M)$ denote its
\emph{optimal ordinal secretary ratio}: the supremum of all $c$ for which
some ordinal policy is $c$-competitive on $M$. In this notation, the strong
conjecture for ordinal algorithms asserts that $\rho(M)\ge1/e$ for every
matroid. A natural starting point for answer this question is to ask how $\rho(M)$ can be computed
for a given matroid. We answer that question in this section via a finite linear
program whose value is exactly $\rho(M)$, and describe the implications we drew from it. The program and the proof of its correctness are in
Appendix~\ref{app:lp}; the proofs concerning the truncation conjecture
discussed at the end of this section are in Appendix~\ref{app:truncation}.

\emph{The program.}
Fix $M$ with $n$ elements. Since policies are ordinal, a weight vector
matters only through the strict total order it induces on $E$ (ties broken
by the fixed priority), which we call the weight order. The observation
behind the program is that, for an ordinal policy, the probability that
after $t$ arrivals the observed set is $S$ and the accepted set is $A$
depends on the hidden weight order only through its restriction to $S$
(Lemma~\ref{lem:ordinal-invariance}). A \emph{state} is therefore a pair
$(\sigma,A)$, where $\sigma$ lists an observed set in decreasing weight
order and $A\subseteq S_\sigma$ is independent. The program has a variable
$p_{\sigma,A}$ for the probability of each state, and variables
$y_{\sigma,e,A}$ and $z_{\sigma,e,A}$ for the probabilities of accepting
and rejecting $e$ when it arrives with prior state $(\sigma-e,A)$. Its
constraints are flow conservation, which says that after a state the next
label is uniform among the unobserved labels and is either accepted or
rejected, and that each state is reached through its possible last
arrivals; feasibility, which forbids dependent acceptances; and
competitiveness, which requires, for every weight order $\pi$ and every
prefix $H$ of $\pi$, that the accepted elements of $H$ number at least
$c\,r(H)$ in expectation. The objective is to maximize $c$. The
competitiveness constraints are the prefix form of $c$-competitiveness,
equivalent to the weighted form for ordinal policies
by~\eqref{eq:layer-cake}. Because a single set of variables is shared by
all weight orders inducing the same relative order on an observed set, one
program encodes the entire adversarial problem at once.

\begin{theorem}[Theorem~\ref{thm:lp-exact}, restated informally]
For every matroid with a nonloop, the value of the program~\eqref{eq:lp} is
$\rho(M)$. Every feasible solution with objective $c$ yields a
$c$-competitive ordinal policy that depends only on the observed weight
order, the arriving label, and the accepted set, and the supremum defining
$\rho(M)$ is attained by such a policy.
\end{theorem}

Linear programming characterizations of secretary policies go back to
Buchbinder, Jain, and Singh~\cite{buchbinder2014}, who treated the classical
and multiple-choice problems, that is, rank one and uniform matroids, where
a policy's situation is summarized by the time and the number of acceptances
so far. For a general matroid the accepted set and the observed weight order
both matter, and the program above indexes states by them; to our knowledge
it is the first formulation that computes $\rho(M)$ for an arbitrary
matroid. Its size grows faster than exponentially in $n$, but symmetries of
$M$ and the merging of states that admit the same future decisions reduce
it substantially without changing its value (Appendix~\ref{sec:lp-reductions}).

Using the catalogue of all matroids on at most nine elements compiled by
Mayhew and Royle~\cite{mayhewroyle2008}, we solved the program to compute
$\rho(M)$ for all matroids of positive rank on seven elements and for
nearly all such matroids on eight elements.\footnote{The implementation and its outputs are available at
\url{https://github.com/TlatoaniHJ/MatroidSecretary}, released as is for
full disclosure.}
In every case the optimal ratio exceeded $1/e$, and furthermore in every case the optimal ratio for a matroid on $n$ elements exceeded the optimal ratio for the rank-$1$ matroid on $n$ elements (with the exception of said matroid itself). This served as significant evidence for the strong conjecture.

We additionally describe other investigations performed based on evidence from the results of the linear program. For $1\le k\le r(M)$, the rank-$k$ \emph{truncation} $M^{(k)}$ of $M$ is the
matroid on the same ground set whose independent sets are the independent
sets of $M$ of size at most $k$; thus $M^{(r(M))}=M$, and $M^{(1)}$ is the
rank-one matroid whose nonloops are the nonloops of $M$. The rank-one
problem is the classical secretary problem, so $\rho(M^{(1)})\ge1/e$. If
$\rho(M^{(k)})$ were nondecreasing in $k$, the strong conjecture would
follow at once: $\rho(M)\ge\rho(M^{(1)})\ge1/e$. The computed values of
$\rho$ on small matroids and their truncations were consistent with this
monotonicity, and we conjectured it.

\begin{conjecture}[Truncation monotonicity; false in general]\label{conj:truncation}
For every matroid $M$ and every $2\le k\le r(M)$,
$\rho(M^{(k)})\ge\rho(M^{(k-1)})$.
\end{conjecture}

The conjecture is true for uniform matroids, where truncation only lowers
the capacity: Theorem~\ref{thm:uniform-monotonicity} shows that
$\rho(U_{k,n})$ is strictly increasing in $k$, which recovers the known
$1/e$ guarantee for uniform matroids~\cite{babaioff2007knapsack}
(Corollary~\ref{cor:uniform-strong}). It is false in general. The
refutation came from our earlier high-girth hardness result for graphic
matroids~\cite{graphic2025}: a graph of girth at least four has rank-three
truncation equal to a uniform matroid, whose ratio exceeds $0.53$, while
there are graphic matroids of arbitrarily large girth whose optimal ratios
are arbitrarily close to $1/e$.
Theorem~\ref{thm:truncation-counterexample} gives a self-contained
counterexample on $K_{2,N}$, and since its upper bound allows the algorithm
to observe numerical weights, the conjecture fails without the ordinal
restriction as well. The failure of this route is part of why the proof of
Section~\ref{sec:linear} controls the accepted span directly rather than
arguing by rank.

\section{A tight guarantee for matroid secretary on linear matroids}\label{sec:linear}

This section proves Theorem~\ref{thm:linear}. We present the algorithm
first, in Section~\ref{sec:linear-algorithm}, as pseudocode built on a
recursively defined family of probability laws and linear programs. The
remainder of the section explains why the algorithm is well defined and why
it has the claimed guarantee. Section~\ref{sec:linear-invariant} introduces
the invariant that controls the accepted set. Section~\ref{sec:linear-chain}
shows how to admit a new element while preserving the invariant along a
chain of subspaces, and how to reduce arbitrary collections of subspaces to
chains by uncrossing. Section~\ref{sec:linear-extension} combines these with
a minimax argument to show that the algorithm's linear programs are always
feasible. Section~\ref{sec:linear-policy} verifies that the online coins
implement the recursion and computes the selection probabilities, and
Section~\ref{sec:linear-cutoff} chooses the cutoff.

We first obtain a representation over a finite field. Rado's
theorem~\cite{rado1957} guarantees that one exists for every finite linear
matroid, including those initially represented over an infinite field.
If no finite-field representation is supplied, enumerate finite fields and
matrices with $\rank(M)$ rows and $n$ columns, testing each matrix against
the independence relation of $M$ on all subsets. This search terminates
under the promise that $M$ is linear and can be performed before any weights
are observed.

Fix the resulting representation $(v_e)_{e\in E}$ in a finite-dimensional
vector space $V$ over a finite field. We may replace $V$ by the span of the
represented vectors. For any subset $A\subseteq E$, write
$W(A)=\spn\{v_e:e\in A\}$, and for a subspace $U\le V$ define the
\emph{load} of $A$ at $U$ as
\[
 \ell_U(A)=\dim(W(A)\cap U).
\]
For a probability law $\mu$ on independent sets and $\beta\ge0$, call $\mu$
\emph{$\beta$-bounded} if
\begin{equation}\label{eq:bounded}
 L_U:=\E_{A\sim\mu}\ell_U(A)\le\beta\dim U
 \qquad\text{for every }U\le V.
\end{equation}
Since $V$ is finite, there are finitely many subspaces and finitely many
independent sets, so~\eqref{eq:bounded} is a finite system of linear
inequalities in the values of $\mu$.

\subsection{The algorithm}\label{sec:linear-algorithm}

Fix a cutoff $1\le s<n$. The algorithm rejects the first $s$ arrivals. For
$t\ge s$, put
\begin{equation}\label{eq:schedule}
 \beta_t=1-\frac st,
 \qquad
 a_t=\frac{s}{t-1}\quad(t>s),
 \qquad\text{so that}\qquad
 \beta_{t-1}+a_t=1
 \quad\text{and}\quad
 \beta_{t-1}+\frac{a_t}t=\beta_t .
\end{equation}
The number $\beta_t$ is the bound the algorithm maintains after $t$
arrivals, and $a_t$ is the probability with which it admits a greedy-basis
candidate arriving at time $t$. In rank one these are the classical cutoff
rule's probability that the slot is already used after $t$ arrivals and its
acceptance probability of a record at time $t$.

For an observed set $S$ with a given relative weight order, the algorithm
uses a law $\mu_S$ on independent subsets of $S$ and, for each $e\in S$, a
\emph{transition} $q^S_e$ that modifies $\mu_{S-e}$ by admitting $e$; the
average over $e\in S$ of the modified laws is $\mu_S$. These are defined
recursively as follows.

\begin{definition}[The laws $\mu_S$ and transitions $q^S$]\label{def:recursion}
Let $S\subseteq E$ with a relative weight order, and let $t=|S|$.
\begin{enumerate}[label=(\roman*)]
\item If $t\le s$, then $\mu_S$ is the point mass at $\varnothing$.
\item If $t>s$, let $B=B(S)$ be the greedy basis of $S$ under the given
order, and assume $\mu_{S-e}$ has been defined for every $e\in S$, with
respect to the weight order restricted to $S-e$. Consider
the \emph{transition program} $(P_S)$ in the variables
$q_e(A)$, for $e\in S$ and independent $A\subseteq S-e$:
\begin{subequations}\label{eq:transition-program}
\begin{alignat}{2}
 & 0\le q_e(A)\le\mu_{S-e}(A),\label{eq:tp-mass}\\
 & q_e(A)=0 &\quad& \text{if } v_e\in W(A) \text{ or } e\notin B,
   \label{eq:tp-feasible}\\
 & \textstyle\sum_A q_e(A)=a_t && \text{for every } e\in B,
   \label{eq:tp-quota}\\
 & g_U(q)\le\beta_t\dim U && \text{for every subspace } U\le V,
   \label{eq:tp-invariant}
\end{alignat}
\end{subequations}
where $g_U(q)$ is the expected load at $U$ after choosing $e\in S$
uniformly and moving mass $q_e(A)$ from $A$ to $A+e$ under $\mu_{S-e}$:
\begin{equation}\label{eq:g-def}
 g_U(q)=\frac1t\sum_{e\in S}\Bigl[
   \sum_A\mu_{S-e}(A)\,\ell_U(A)
   +\sum_A q_e(A)\bigl(\ell_U(A+e)-\ell_U(A)\bigr)\Bigr].
\end{equation}
Let $q^S$ be a canonical feasible solution of $(P_S)$, for example the
lexicographically smallest one. Define $\mu_S$ as the resulting average law:
for independent $A'\subseteq S$,
\begin{equation}\label{eq:mu-def}
 \mu_S(A')=\frac1t\sum_{e\in S}
 \begin{cases}
  \mu_{S-e}(A')-q^S_e(A') & e\notin A',\\
  q^S_e(A'-e) & e\in A'.
 \end{cases}
\end{equation}
\end{enumerate}
\end{definition}

Constraint~\eqref{eq:tp-mass} says that $q_e$ moves part of the mass that
$\mu_{S-e}$ places at $A$ to $A+e$; \eqref{eq:tp-feasible} forbids dependent
additions and admits only greedy-basis elements; \eqref{eq:tp-quota} says
that each greedy-basis element is admitted with total probability $a_t$; and
\eqref{eq:tp-invariant} asks that the averaged law remain $\beta_t$-bounded.
The definition makes sense only if $(P_S)$ is feasible and each $\mu_{S-e}$
is a probability law; both are established in
Section~\ref{sec:linear-extension}. The program is finite, with rational
data, since the weights enter only through their order.

\begin{algorithm}[htbp]
\caption{Ordinal secretary policy for a linear matroid}\label{alg:linear}
\DontPrintSemicolon
\KwIn{A representation $(v_e)_{e\in E}$ of $M$ over a finite field; a cutoff $1\le s<n$.}
Initialize $A\gets\varnothing$\;
\For{$t=1,\ldots,n$, with arriving element $e$ and observed set $S$ (so $|S|=t$)}{
  \lIf{$t\le s$}{reject $e$}
  \Else{
    Compute $\mu_{S-e}$ and $q^S$ from the observed relative weight order on $S$, by Definition~\ref{def:recursion}\;
    With probability $q^S_e(A)/\mu_{S-e}(A)$ (zero if $\mu_{S-e}(A)=0$), accept $e$ and set $A\gets A+e$; otherwise reject $e$\;
  }
}
\Return{$A$}
\end{algorithm}

Algorithm~\ref{alg:linear} is ordinal: the recursion uses only the relative
weight order on subsets of $S$ and the fixed representation, so all its
inputs are available when $e$ arrives. It is well defined by
Proposition~\ref{prop:feasible} below, which also shows that the acceptance
probability lies in $[0,1]$ and that accepted sets are independent. The
recursion can be evaluated exactly: all subspaces and independent sets can
be enumerated over the finite field, the programs have rational
coefficients, and canonical rational solutions can be chosen. The number of
states and constraints can be enormous, which is why no running-time bound
is claimed.

\begin{proposition}[Well-definedness]\label{prop:feasible}
For every $S$ with relative weight order and $|S|\ge s$, the law $\mu_S$ is
a probability law on independent subsets of $S$ and is
$\beta_{|S|}$-bounded. For $|S|>s$ the program $(P_S)$ is feasible.
\end{proposition}

\begin{proposition}[Selection probabilities]\label{prop:selection}
Under Algorithm~\ref{alg:linear}, for every weight vector $w$ and every
$e\in B(E)$, conditional on the observed set $S$ containing $e$ with $e$
arriving last at time $t>s$, the element $e$ is accepted with probability
exactly $a_t$. Consequently $\Pr(e\in A)=c_n(s)$.
\end{proposition}

Theorem~\ref{thm:linear} follows from these two propositions and the choice
of cutoff in Section~\ref{sec:linear-cutoff}. The recursion is, in the
language of Appendix~\ref{app:lp}, an explicit feasible solution of the
per-element version of the ordinal secretary program, built one observed set
at a time.

\emph{The algorithm in rank one.}
If $M$ has rank one and no loops, then $V$ is a line, its only subspaces
are $0$ and $V$,
the greedy basis of $S$ is its heaviest element, and an accepted set is
either empty or a singleton whose span is $V$. We claim that
$\mu_S(\varnothing)=1-\beta_t=s/t$ for $|S|=t\ge s$. This holds for $t=s$.
For $t>s$, the only row of $(P_S)$ with a nonzero quota is that of the
heaviest element $e$ of $S$, and its quota $a_t=s/(t-1)$ equals
$\mu_{S-e}(\varnothing)$ by induction, so~\eqref{eq:tp-quota} forces
$q^S_e(\varnothing)=\mu_{S-e}(\varnothing)$; the constraint at $V$ then
holds with equality by~\eqref{eq:schedule}, and~\eqref{eq:mu-def} gives
$\mu_S(\varnothing)=\frac1t\bigl[(t-1)\tfrac{s}{t-1}+0\bigr]=s/t$.
Algorithm~\ref{alg:linear} therefore accepts an arrival with probability one
exactly when $t>s$, the arrival is a record, and nothing has been accepted:
this is the classical cutoff rule, and $\beta_t=1-s/t$ is the probability
that it has accepted something by time $t$. The subspace invariant is what
replaces the single number $\beta_t$ in higher rank.

\subsection{The subspace invariant}\label{sec:linear-invariant}

We first give intuition as to why it is necessary to control the load at every subspace, rather than simply those of dimension $1$. For a nonzero vector $v$, the constraint at the line $\langle v\rangle$ gives
\begin{equation}\label{eq:availability}
 \Pr(v\in W(A))=\E\ell_{\langle v\rangle}(A)\le\beta,
\end{equation}
so at least $1-\beta$ of the probability mass permits adding $v$. By the
schedule~\eqref{eq:schedule}, after $t-1$ arrivals this is exactly $a_t$,
the mass we wish to spend on a candidate arriving at time $t$. Lines of
ground vectors would suffice for availability alone. They do not suffice to
propagate the invariant, and the constraints on spans of ground elements
are strictly weaker than~\eqref{eq:bounded}. For example, let $u_1,u_2,u_3$
be a basis of $\F_2^3$, take ground vectors $u_1,u_2,u_1+u_3,u_2+u_3$, and
choose $A$ uniformly from the first and last pairs. The expected loads of
ground-generated lines and planes are at most $1/2$ and $3/2$, respectively,
while $V$ has load $2$. These constraints hold with $\beta=3/4$, yet the
common line $\langle u_1+u_2\rangle$ of the two possible spans has load $1$.
The proofs below use the constraint at $U+\langle v_e\rangle$ to control the
cost of an addition at $U$, and such subspaces are in general not generated
by ground vectors.

Two identities describe how a load changes under one-dimensional
extensions:
\begin{align}
 \dim((W+\langle v\rangle)\cap U)-\dim(W\cap U)
   &=\ind\{v\in W+U\}, &&v\notin W,\label{eq:load-increment}\\
 \dim(W\cap(U+\langle v\rangle))-\dim(W\cap U)
   &=\ind\{v\in W+U\}, &&v\notin U.\label{eq:adjacent-load}
\end{align}
Both follow by expressing an intersection dimension as the sum of the
two dimensions minus the dimension of their sum. In the first identity,
adjoining $v$ increases $\dim W$ by one; in the second, it increases
$\dim U$ by one. The corresponding sum increases in dimension exactly
when $v\notin W+U$. Identity~\eqref{eq:load-increment} says that accepting
$e$ raises the load at $U$ by one precisely at states with
$v_e\in W(A)+U$; we call such states \emph{expensive} for $U$.
Identity~\eqref{eq:adjacent-load} says that, when $v_e\notin U$, the
probability of an expensive state is the difference between the loads at
$U+\langle v_e\rangle$ and at $U$, which the invariant bounds.

\subsection{Admitting an element along a chain}\label{sec:linear-chain}

An \emph{addition rule} for a represented element $e$ specifies masses
$q(A)$ with $0\le q(A)\le\mu(A)$, and $q(A)=0$ whenever $v_e\in W(A)$.
It retains mass $\mu(A)-q(A)$ at $A$ and moves mass $q(A)$ to $A+e$.
Splitting an atom in this way is implemented by a randomized decision
conditional on that accepted set; this is what Algorithm~\ref{alg:linear}
does. The first lemma handles a chain of subspaces, where a single order of
the states is cheapest for all members at once.

\begin{lemma}[Chain admission]\label{lem:chain}
Let $\mu$ be $\beta$-bounded, let $v_e\ne0$, and suppose $a\ge0$ and
$\beta+a\le1$. For every finite chain of subspaces
$C_1\subseteq\cdots\subseteq C_m$, there is a feasible addition rule of
total mass $a$ whose updated law satisfies
$\E\ell_U(A')\le\beta\dim U+a\ind\{v_e\in U\}$ for every subspace
$U$ in the chain.
\end{lemma}

\begin{proof}
For each chain member $U$, let $D_U=\{A:v_e\in W(A)+U\}$ be the set of
expensive states, and let $p_U=\mu(D_U)$. If $U\subseteq U'$ then
$D_U\subseteq D_{U'}$, so the sets $D_{C_1}\subseteq\cdots\subseteq D_{C_m}$
are nested. Every infeasible state, one with $v_e\in W(A)$, lies in every
$D_U$. Assign each feasible state $A$ the index
$i(A)=\min\{i:A\in D_{C_i}\}$, with $i(A)=m+1$ if $A$ lies in no $D_{C_i}$.
Fill mass $a$ into feasible states in decreasing order of $i(A)$, splitting
the last atom if needed, so that the rule uses states outside $D_{C_j}$
before any state inside $D_{C_j}$, simultaneously for every $j$. Such a rule
exists because the feasible mass is at least $1-\beta\ge a$
by~\eqref{eq:availability}.

Fix a chain member $U$. Every state outside $D_U$ is feasible, and the rule
fills those states first; their total mass is $1-p_U$. By
\eqref{eq:load-increment}, the load at $U$ increases only on the mass
placed inside $D_U$, which is exactly $(a-(1-p_U))_+=(a-1+p_U)_+$. If
$v_e\notin U$, equation~\eqref{eq:adjacent-load} gives
$p_U=L_{U+\langle v_e\rangle}-L_U$, and the invariant at $U$ and at
$U+\langle v_e\rangle$ gives $L_U\le\beta\dim U$ and
$L_U+p_U\le\beta(\dim U+1)$. Consequently,
\[
 L_U+(a-1+p_U)_+
 =\max\{L_U,\;L_U+p_U+a-1\}\le\beta\dim U,
\]
where the last inequality uses $\beta+a\le1$. If $v_e\in U$, then
$D_U$ is everything, so every feasible addition increases the load by one
and the expected increase is $a$.
\end{proof}

The chain provides a common ordering of the feasible states for all its
constraints. An arbitrary collection of subspaces need not admit such an
ordering. We next show that a priced collection can be dominated by a
chain without increasing its dimension budget. Figure~\ref{fig:uncrossing}
illustrates the step.

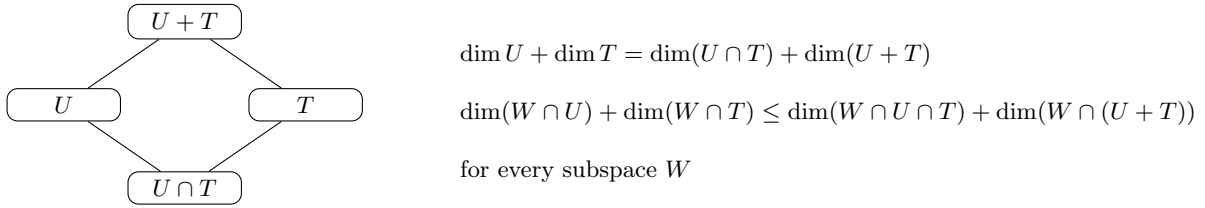
\begin{figure}[htbp]
\centering
\begin{tikzpicture}[x=1.6cm,y=1.1cm,
  every node/.style={font=\small},
  sub/.style={draw,rounded corners,inner sep=3pt,minimum width=1.5cm}]
 \node[sub] (UT) at (0,2) {$U+T$};
 \node[sub] (U) at (-1,1) {$U$};
 \node[sub] (T) at (1,1) {$T$};
 \node[sub] (UcT) at (0,0) {$U\cap T$};
 \draw (UcT) -- (U) -- (UT);
 \draw (UcT) -- (T) -- (UT);
 \node[align=left,anchor=west] at (2.2,1.6)
  {$\dim U+\dim T=\dim(U\cap T)+\dim(U+T)$};
 \node[align=left,anchor=west] at (2.2,0.9)
  {$\dim(W\cap U)+\dim(W\cap T)\le\dim(W\cap U\cap T)+\dim(W\cap(U+T))$};
 \node[align=left,anchor=west] at (2.2,0.2)
  {for every subspace $W$};
\end{tikzpicture}
\caption{Uncrossing two incomparable subspaces $U,T$ into $U\cap T$ and
$U+T$. The dimension budget is preserved and the load of every $W$ can only
increase, so a priced family of constraints is dominated by a chain.}
\label{fig:uncrossing}
\end{figure}

\begin{lemma}[Subspace uncrossing]\label{lem:uncrossing}
For nonnegative rational prices $(\lambda_U)_{U\le V}$, there are a
chain $C_1\subseteq\cdots\subseteq C_m$ and nonnegative rational
coefficients $\gamma_1,\ldots,\gamma_m$ such that, for every $W\le V$,
\begin{equation}\label{eq:chain-domination}
 \sum_U\lambda_U\dim(W\cap U)
 \le\sum_i\gamma_i\dim(W\cap C_i),
 \qquad
 \sum_U\lambda_U\dim U=\sum_i\gamma_i\dim C_i.
\end{equation}
\end{lemma}

\begin{proof}
Clear denominators and regard the prices as a multiset of subspaces.
Whenever two members $U,T$ are incomparable, replace them by $U\cap T$
and $U+T$. The dimension formula preserves the total dimension. Moreover,
\[
 \dim(W\cap U)+\dim(W\cap T)
 \le \dim(W\cap U\cap T)+\dim(W\cap(U+T))
\]
for every $W$: the intersection of $W\cap U$ and $W\cap T$ is
$W\cap U\cap T$, and their sum is contained in $W\cap(U+T)$.
Thus the replacement increases the sum of loads pointwise.

The process terminates. If $d=\dim(U\cap T)$, replacing incomparable
$U,T$ increases the sum of squared dimensions by
$2(\dim U-d)(\dim T-d)>0$. This integer potential is bounded by the
fixed multiset size times $(\dim V)^2$. At termination all members
are comparable. Dividing their multiplicities by the original common
denominator proves~\eqref{eq:chain-domination}.
\end{proof}

\subsection{Simultaneous extension}\label{sec:linear-extension}

We now show that the transition program $(P_S)$ is feasible. The lemma is
stated for arbitrary input laws $\mu_e$, one for each possible identity $e$
of the last arrival; no common coupling between them is assumed. The
conclusion concerns the average of the updated laws over $e$, which is what
a uniformly random last arrival produces.

\begin{lemma}[Averaged extension]\label{lem:extension}
Let $S\subseteq E$ have size $t\ge1$, and let $B\subseteq S$ be independent.
For each $e\in S$, let $\mu_e$ be a $\beta$-bounded law on independent
subsets of $S-e$. If $0\le a\le1-\beta$, there are feasible addition
rules admitting mass $a$ in each row $e\in B$ and zero in every other
row, such that the average of the updated laws is
$(\beta+a/t)$-bounded.
\end{lemma}

\begin{proof}
Let $P$ consist of the transition arrays $q=(q_e(A))$ satisfying
$0\le q_e(A)\le\mu_e(A)$, with zero mass for infeasible additions
and total mass $a\ind\{e\in B\}$ in each row $e$; these are
constraints~\eqref{eq:tp-mass}--\eqref{eq:tp-quota} with $\mu_e$ in place
of $\mu_{S-e}$. Availability~\eqref{eq:availability} with $a\le1-\beta$
shows that $P$ is nonempty. It is a compact convex polytope. For $q\in P$,
write $g_U(q)$ for the expected load at $U$ after choosing a row uniformly
and making its prescribed update, as in~\eqref{eq:g-def}. Each $g_U$ is
affine in $q$. We must find $q\in P$ with $g_U(q)\le(\beta+a/t)\dim U$
for every $U$.

The strategy is a minimax argument: we show that every nonnegative
combination of the constraints can be satisfied by some $q\in P$, and
conclude by separation that a single $q$ satisfies all of them. Fix
rational nonnegative prices $\lambda_U$, and write
$R=\sum_U\lambda_U\dim U$. Uncross them using
Lemma~\ref{lem:uncrossing} into a chain $C_1\subseteq\cdots\subseteq C_m$
with coefficients $\gamma_i$. Apply Lemma~\ref{lem:chain} to this chain in
every row $e\in B$ and leave the other rows unchanged. For each chain
member $C_i$, the chain lemma bounds the updated load in row $e$ by
$\beta\dim C_i+a\ind\{v_e\in C_i\}$, so averaging over the $t$ rows and
using the independence of $B$,
\[
 g_{C_i}(q)\le\beta\dim C_i+
    \frac at|\{e\in B:v_e\in C_i\}|
 \le\Bigl(\beta+\frac at\Bigr)\dim C_i,
\]
since an independent set has at most $\dim C_i$ vectors in $C_i$.
Pointwise domination~\eqref{eq:chain-domination} and preservation of the
dimension budget now give $\sum_U\lambda_U g_U(q)\le(\beta+a/t)R$. The
choice of $q$ depends on the prices, which is enough for the next step.

The finite set of subspaces indexes an excess vector
$z(q)=(g_U(q)-(\beta+a/t)\dim U)_U$. If no $q\in P$ satisfied all the
constraints, the compact convex set $\{z(q):q\in P\}$ would be
disjoint from the closed negative orthant. Strict separation would
give a nonnegative real vector $\lambda$ with
$\min_{q\in P}\lambda\cdot z(q)>0$. The vector is nonnegative because
the orthant is unbounded in each negative coordinate direction. The
excess vectors are bounded, so a sufficiently close rational nonnegative
vector would still give a strictly positive minimum, contradicting the
priced bound just proved. Hence one $q$ satisfies every constraint.
\end{proof}

The lemma does not assert that each updated row is bounded by the new
parameter. It is the uniform average that is controlled. This distinction
allows the independence of the whole candidate set $B$ to pay for the
increase, through $|\{e\in B:v_e\in U\}|\le\dim U$.

\begin{proof}[Proof of Proposition~\ref{prop:feasible}]
We argue by induction on $|S|=t\ge s$. For $t=s$, $\mu_S$ is the point
mass at $\varnothing$, which is a probability law and is $0$-bounded, and
$\beta_s=0$. Let $t>s$ and assume the claim for all sets of size $t-1$.
Then each $\mu_{S-e}$ is a $\beta_{t-1}$-bounded probability law on
independent subsets of $S-e$. Apply Lemma~\ref{lem:extension} with
$B=B(S)$, input laws $\mu_e=\mu_{S-e}$, parameter $\beta=\beta_{t-1}$, and
$a=a_t$; the hypothesis $a\le1-\beta$ holds with equality
by~\eqref{eq:schedule}. The addition rules it produces form a point of
$(P_S)$: constraints~\eqref{eq:tp-mass}--\eqref{eq:tp-quota} are the
definition of $P$, and~\eqref{eq:tp-invariant} is the conclusion, since
$\beta_{t-1}+a_t/t=\beta_t$. Thus $(P_S)$ is feasible. For any feasible
$q^S$, formula~\eqref{eq:mu-def} defines the average of the updated row
laws, each of which is a probability law on independent subsets of $S$ by
\eqref{eq:tp-mass} and~\eqref{eq:tp-feasible}; so $\mu_S$ is a probability
law, and~\eqref{eq:tp-invariant} says it is $\beta_t$-bounded.
\end{proof}

\subsection{Analysis of the policy}\label{sec:linear-policy}

We now verify that the coins of Algorithm~\ref{alg:linear} implement the
recursively defined laws, and compute the selection probabilities.

\begin{lemma}\label{lem:implementation}
Fix a weight vector $w$ and a set $S$ with $|S|=t\ge s$. Conditional on the
first $t$ arrivals forming $S$, the accepted set of
Algorithm~\ref{alg:linear} after these arrivals has law $\mu_S$, where
$\mu_S$ is computed from the relative weight order that $w$ induces on $S$.
\end{lemma}

\begin{proof}
For $t=s$ nothing has been accepted and $\mu_S$ is the point mass at
$\varnothing$. Let $t>s$. Conditional on the first $t$ labels forming $S$,
their last label $e$ is uniform in $S$; conditional also on that label, the
preceding $t-1$ labels form $S-e$ in a uniform order. By induction, the
accepted set before this arrival has law $\mu_{S-e}$. The coin in
Algorithm~\ref{alg:linear} accepts with probability
$q^S_e(A)/\mu_{S-e}(A)$ from each state $A$ of positive mass, which moves
exactly mass $q^S_e(A)$ from $A$ to $A+e$; constraint~\eqref{eq:tp-mass}
makes this a probability, and~\eqref{eq:tp-feasible} keeps the accepted set
independent. The updated law in row $e$ is therefore the one
in~\eqref{eq:mu-def}, and averaging over the uniform $e$ gives $\mu_S$.
Conditioning on the identity of the next label does not alter the previous
internal coins, since they were independent of the arrival order beyond
the observed set.
\end{proof}

\begin{proof}[Proof of Proposition~\ref{prop:selection}]
Let $e\in B(E)$, the greedy basis of the whole ground set. By the second
greedy property of Section~\ref{sec:preliminaries}, $e\in B(S)$ for every
$S\ni e$, so $e$ is a candidate in every observed set containing it.
Condition on the observed set $S$ with $e$ arriving last at time $t>s$.
By Lemma~\ref{lem:implementation}, the accepted set before this arrival
has law $\mu_{S-e}$, and the coin moves total mass
$\sum_A q^S_e(A)=a_t$ by~\eqref{eq:tp-quota}. Thus $e$ is accepted with
conditional probability exactly $a_t$. The arrival position of $e$ is
uniform on $\{1,\ldots,n\}$ and independent of the set of labels preceding
it, so
\[
 \Pr(e\in A)=\frac1n\sum_{t=s+1}^n a_t
 =\frac1n\sum_{t=s+1}^n\frac{s}{t-1}
 =\frac sn\sum_{k=s}^{n-1}\frac1k=c_n(s).
\qedhere
\]
\end{proof}

Summing the weights of the elements of $B(E)$ shows that
$\E w(A)\ge c_n(s)\,w(B(E))=c_n(s)\OPT(w)$, which is the fixed-cutoff claim
of Theorem~\ref{thm:linear}. As noted in the introduction, the same
computation gives the success probability of the classical cutoff rule; the
algorithm reproduces it for every element of the optimum simultaneously.

\subsection{Choosing the cutoff}\label{sec:linear-cutoff}

We finish with a finite-$n$ bound, so the theorem does not rely on an
asymptotic choice of $s$. For $n=2$, the cutoff $s=1$ gives $c_2(1)=1/2$.
For $n\ge3$, put $z=n/e$, $r=\lfloor z\rfloor$, and $\theta=z-r$;
both $r$ and $r+1$ are allowed cutoffs. Using $\log u\le u-1$ first on
$(r+1)/z$ and then on $(k+1)/k$ gives
\begin{align*}
 \frac1e=\frac zn\log\frac nz
 &\le\frac{r+1-z}{n}+\frac zn\sum_{k=r+1}^{n-1}\frac1k\\
 &= (1-\theta)c_n(r)+\theta c_n(r+1).
\end{align*}
The last equality follows by separating the term $1/r$ in $c_n(r)$.
At least one of these two cutoffs therefore attains $1/e$.
This completes the proof of Theorem~\ref{thm:linear}.

\section{A tight single-sample prophet inequality}\label{sec:prophet}

We now consider an arbitrary matroid and prove Theorem~\ref{thm:prophet}.
As in Section~\ref{sec:linear}, we present the algorithm first
(Section~\ref{sec:prophet-algorithm}) and then its analysis
(Sections~\ref{sec:prophet-invariant} and~\ref{sec:prophet-guarantee}).
The algorithm maintains a
stored vector $W$, initially equal to the sample vector, and keeps its
accepted set inside $B(W)$. When an actual value arrives, the algorithm
may replace its coordinate in $W$. The replacement and acceptance rules
are chosen together, so that the stored vector retains the original
product distribution and the accepted set is a \emph{fair thinning} of the
greedy basis of $W$: each processed element of that basis is retained
independently with probability $1/2$.

The conclusion is stronger than Theorem~\ref{thm:prophet}. If $Y$ is
a fresh draw from the product value distribution, and $T$ independently
retains each element of $B(Y)$ with probability $1/2$, then our output
satisfies
\begin{equation}\label{eq:marked-identity}
 \{(e,X_e):e\in A\}\dist\{(e,Y_e):e\in T\}.
\end{equation}
Taking expected total value in~\eqref{eq:marked-identity} gives
$\E\sum_{e\in A}X_e=\tfrac12\E\sum_{e\in B(Y)}Y_e
=\tfrac12\E\OPT(Y)=\tfrac12\E\OPT(X)$.

\subsection{Algorithm}\label{sec:prophet-algorithm}

We first recall how a greedy basis changes when one coordinate changes.
We include zero-weight elements in greedy throughout this section, so
every greedy set is a full basis.

\begin{lemma}\label{lem:sensitivity}
Fix all coordinates other than $e$, and let $B_0,B_1$ be the greedy
bases when the value of $e$ is respectively $v_e^-\le v_e^+$. If $e$
belongs to both bases or to neither, then $B_0=B_1$. Otherwise,
$e\notin B_0$, $e\in B_1$, and $B_1=B_0-f+e$ for one element $f$.
\end{lemma}

\begin{proof}
Greedy accepts a label exactly when it is not in the closure of the labels
scanned before it. Raising the value of $e$ from $v_e^-$ to $v_e^+$ moves
$e$ earlier in the scan: its set of predecessors shrinks, while the set of
predecessors of every other label either gains $e$ or is unchanged. By the
closure criterion, $e$ can only enter the basis, and every other label can
only leave it. Both bases have size $r(E)$. Hence if $e$ enters, exactly
one other label $f$ leaves, and if $e$ does not enter, nothing changes.
\end{proof}

Let $P$ be the set of processed labels and $A$ the accepted set. When $e$
arrives with actual value $x=X_e$, compare the current basis $C=B(W)$ with
the basis $D$ obtained by replacing $W_e$ with $x$. A \emph{refresh} of
coordinate $e$ means making this replacement in the stored vector, that
is, setting $W_e=x$. If $e\in D\setminus C$,
Lemma~\ref{lem:sensitivity} gives $D=C-f+e$ for a unique label $f$.
The algorithm \emph{refreshes} with
probability
\begin{equation}\label{eq:refresh-rule}
 p=\begin{cases}
  1/2,& e\in C,\\
  0,& e\notin C\cup D,\\
  0,& e\in D\setminus C,\ f\in A,\\
  1,& e\in D\setminus C,\ f\in P\setminus A,\\
  1/2,& e\in D\setminus C,\ f\notin P,
 \end{cases}
\end{equation}
using a fresh independent fair coin whenever $p=1/2$. On a refresh it sets
$W_e=x$ and accepts $e$ if and only if $e\in D$; without a refresh it keeps
$W$ and rejects $e$. Then $e$ is added to $P$. Algorithm~\ref{alg:prophet}
states the same rule as pseudocode.

The approach of repeatedly recomputing the current optimum and implementing its
recommendation only when compatible with previous commitments is reminiscent of
the random-order matching framework of Kesselheim et al.~\cite{kesselheim2013}.
The intuition behind the cases is as follows. If $e$ is already in the basis on the
strength of its sample, a fair coin decides whether to replace the sample
by the actual value, and $e$ is accepted when it survives the replacement.
If $e$ newly enters and displaces $f$, an accepted $f$ blocks the exchange,
and an already rejected $f$ permits it with probability one. If $f$ has not
arrived, a fresh fair coin is used.

\begin{algorithm}[htbp]
\caption{Single-sample matroid prophet algorithm}\label{alg:prophet}
\DontPrintSemicolon
Initialize $W\gets S$, $A\gets\varnothing$, $P\gets\varnothing$, and $C\gets B(W)$\;
\For{each arriving label $e$ with actual value $x=X_e$}{
  Let $W'$ equal $W$ with coordinate $e$ replaced by $x$, and compute $D\gets B(W')$\;
  \uIf{$e\in C$}{
    With probability $1/2$, set $W\gets W'$ and $C\gets D$, and add $e$ to $A$ if $e\in D$\;
  }
  \ElseIf{$e\in D$}{
    Let $f$ be the unique element of $C\setminus D$\;
    Set $p\gets0$ if $f\in A$, $p\gets1$ if $f\in P\setminus A$, and $p\gets1/2$ if $f\notin P$\;
    With probability $p$, set $W\gets W'$, $C\gets D$, and $A\gets A+e$\;
  }
  Set $P\gets P+e$\;
}
\Return{$A$}
\end{algorithm}

\emph{The rule in rank one.}
To see the rule in the simplest case, let $M$ have rank one with no loops,
so that $B(W)$ is the single element holding the largest stored value, with
ties broken by priority. Initially the stored values are the samples and
$C=\{c\}$, where $c$ holds the largest sample. When $e$ arrives, there are
two ways for it to matter. If $e=c$, so that its own sample is the current
maximum, a fair coin decides whether to overwrite that sample by the actual
value $X_e$; on a refresh, $e$ is accepted if it still holds the maximum.
If $e\ne c$ and replacing its stored value by $X_e$ makes it the maximum, accepting $e$ would displace $c$: if
$c$ was accepted earlier, $e$ is rejected; if $c$ arrived earlier and was
rejected, $e$ is accepted; and if $c$ has not yet arrived, $e$ is accepted
with probability $1/2$. The rule is randomized even in rank one, and it
differs from the maximum-sample threshold rule of
Rubinstein, Wang, and Weinberg~\cite{rubinstein2020}, which accepts the
first value above the threshold (with independent random tie-breakers
when distributions have atoms). The coins are what make the
output an exact fair thinning of a fresh optimum, and tracking the displaced
element is what allows the same rule to work in every matroid.

\begin{lemma}\label{lem:accepted-basis}
After every arrival processed by Algorithm~\ref{alg:prophet}, the accepted
set and stored vector satisfy $A\subseteq B(W)$ and $W_e=X_e$ for all
$e\in A$.
\end{lemma}

\begin{proof}
The arriving label has not previously been accepted. If $e\in C$, a refresh
removes no other old basis member, by Lemma~\ref{lem:sensitivity}; if $e$
enters $D$ from outside $C$, the rule~\eqref{eq:refresh-rule} forbids
displacing an accepted $f$. Acceptance always refreshes the accepted
coordinate to its actual value, and no coordinate changes after its own
arrival. Induction proves both claims; in particular the output is
independent on every realization.
\end{proof}

\subsection{The conditional thinning invariant}\label{sec:prophet-invariant}

The sample and the actual value of a label are two independent draws from
the same distribution. Their joint law can be generated by first drawing
the unordered pair, then using an independent fair bit to assign its two
endpoints to the sample and actual value. These orientation bits are
independent across labels, including when equal endpoints are distinguished
only for the analysis.
Moreover, the stored coordinate $W_e$ is always one of the two draws, so the
algorithm's state on coordinate $e$ is captured by a single bit recording
which draw is stored.

Fix $v_e^-\le v_e^+$, and independently assign one endpoint as
the sample and the other as the actual value by a fair bit. Keep the arrival order
fixed independently of these assignments. The algorithm sees the sample
endpoints at the start and the actual endpoint of each label at its
arrival, but does not see the pairs in advance.

Introduce a latent bit $Z_e$ recording the endpoint currently stored in
$W_e$: bit zero denotes $v_e^-$ and bit one denotes $v_e^+$. Initially
these bits are independent and fair. A refresh flips the current bit,
even if its endpoints are numerically equal. Before $e$ arrives, its
actual endpoint is opposite its stored endpoint. These bits are only
proof variables; the algorithm does not use them.

For $z\in\{0,1\}^E$, let $W(z)$ be the vector with
$W(z)_e=v_e^-$ if $z_e=0$ and $W(z)_e=v_e^+$ if $z_e=1$.
Thus $B(W(Z))$ is the greedy basis of the current stored vector,
with the fixed tie priority.

\begin{lemma}\label{lem:joint-thinning}
After every fixed arrival prefix $P$, the vector $Z$ is uniform on
$\{0,1\}^E$. Conditional on the entire vector $Z$, the accepted set is
an independent fair thinning of $B(W(Z))\cap P$.
\end{lemma}

\begin{proof}
Initially $P=\varnothing$, so the claim holds. Suppose it holds
before $e$ arrives. Fix $(Z_g)_{g\ne e}$, and let $B_i$ be the resulting
basis $B(W(Z))$ when $Z_e=i$, for $i\in\{0,1\}$. By the induction
hypothesis, $\Pr(Z_e=i\mid (Z_g)_{g\ne e})=1/2$. Write
$I_g=\mathbf 1_{\{g\in A\}}$ for each $g\in P$. Conditional on $Z_e=i$,
the variables $(I_g)_{g\in B_i\cap P}$ are independent, and each equals
one with probability $1/2$. Let $H=B_0\cap B_1\cap P$. The distribution
of $(I_g)_{g\in H}$ is the same for $i=0$ and $i=1$; when an exchanged
element $f\in P$ exists, these variables are also independent of $I_f$
conditional on each value of $Z_e$. Thus we can fix $(I_g)_{g\in H}$
without changing the probabilities for $Z_e$ or $I_f$ used below.
The update does not change any $I_g$ with $g\in H$.
We now compute the joint probabilities of $Z_e$ and $I_e$
after processing $e$, together with $I_f$ when $f\in P$.

If neither basis contains $e$, the bases coincide and nothing changes.
If both contain $e$, they again coincide. For each value of $Z_e$
before processing $e$, a fair coin either refreshes the coordinate,
flipping $Z_e$ and accepting $e$, or leaves $Z_e$ unchanged and rejects
$e$. After processing $e$, for each $z\in\{0,1\}$, the joint
probabilities $\Pr(Z_e=z,I_e=1)$ and $\Pr(Z_e=z,I_e=0)$ are each $1/4$.
Thus $\Pr(I_e=1\mid Z_e=z)=1/2$.

In the remaining case write $B_1=B_0-f+e$. If $f\notin P$, either
orientation flips with probability $1/2$, and only a flip from zero
to one accepts. The output-zero mass is $1/2$, with $e$ outside the
basis; the output-one masses with $e$ accepted and unaccepted are
each $1/4$. There is no processed indicator for $f$ to preserve.

Suppose instead that $f\in P$. There are three input states after
the common indicators have been suppressed. The zero-bit state with
$f$ accepted has mass $1/4$ and stays unchanged. The zero-bit state
with $f$ unaccepted has mass $1/4$ and surely flips, accepting $e$.
The one-bit state has mass $1/2$; its fair refresh coin sends half
this mass to bit zero with $f$ unaccepted, and leaves half at bit
one with $e$ unaccepted. Thus the joint output masses are
\begin{center}
\begin{tabular}{lcc}
\toprule
Stored bit and processed basis member & Accepted & Unaccepted\\
\midrule
$Z_e=0$, member $f$ & $1/4$ & $1/4$\\
$Z_e=1$, member $e$ & $1/4$ & $1/4$\\
\bottomrule
\end{tabular}
\end{center}
Each output bit is fair, and its processed basis member has an
independent fair accepted indicator. Multiplying by the unchanged
product law of the common indicators proves the conditional thinning
claim. For every configuration of the other bits, the output bit
is fair, so the entire vector remains uniform. This completes the
induction.
\end{proof}

\subsection{Distributional guarantee and optimality}\label{sec:prophet-guarantee}

\begin{proof}[Proof of Theorem~\ref{thm:prophet}]
For fixed endpoint pairs, Lemma~\ref{lem:joint-thinning} says that at
termination $W$ has the product law of fair endpoint choices. Conditional
on $W$, $A$ is an independent fair thinning of $B(W)$. By
Lemma~\ref{lem:accepted-basis}, each selected stored value is its actual
value. This proves
the fixed-pair version of~\eqref{eq:marked-identity}, with $Y$ a fresh
independent fair choice of endpoints.

For general distributions, draw two independent values per label
and condition on their unordered pairs. Independent fair orientations
recover the law of the sample and actual vectors, including when the
two values coincide. Averaging the fixed-pair result over the pairs
proves~\eqref{eq:marked-identity}.

Finally, nonnegativity and greedy optimality give
\[
 \E\sum_{e\in A}X_e=\E\sum_{e\in T}Y_e
 =\frac12\,\E\sum_{e\in B(Y)}Y_e=\frac12\,\E\OPT(X),
\]
with no integrability assumption.

All data used at a step are the initially given samples and values
already revealed. The algorithm caches the current basis and computes
one replacement basis per arrival, for a total of at most $n+1$ greedy
scans. Each scan uses at most $n$ independence queries. Thus the
algorithm is online and uses $O(n^2)$ independence queries, proving
the algorithmic part of Theorem~\ref{thm:prophet}.

The tightness is well-known and we provide only for completeness. Consider a rank-one matroid with two elements arriving
in the order $1,2$. Let $X_1=1$, and let $X_2$ be $H$ with probability
$1/H$ and zero otherwise. Even an algorithm knowing these distributions
has expected reward at most one: accepting the first value gives one,
and rejecting leaves an independent future reward of mean one. Extra
independent samples cannot change this comparison. The prophet obtains
$2-1/H$ in expectation, so no uniform factor greater than $1/2$ is
possible as $H\to\infty$.
\end{proof}

\section{Secretary guarantees from samples}\label{sec:reduction}

In this section, we describe a reduction from the matroid secretary problem to
the single-sample matroid prophet that preserves constant-competitiveness. The
existence of this reduction is originally credited to Wenzheng Li by Fu et
al.~\cite{fu2024}, who cite private communication.\footnote{The second author thanks Neel Patel for informing him of the existence of this reduction.}
As we do not have access to
this private communication, we, for completeness, give a self-contained
quantitative
version, which yields Corollary~\ref{cor:general-secretary} when applied
to our single-sample algorithm.

\begin{theorem}[Reduction to secretary]\label{thm:reduction}
Let $0<\alpha\le1$. An $\alpha$-competitive single-sample prophet
algorithm for arbitrary matroids yields an
$\alpha^2/16$-competitive secretary algorithm for arbitrary matroids.
It suffices for the prophet algorithm to handle a uniformly random arrival
order, independent of the samples and values, without knowing the
permutation in advance. The reduction preserves polynomial running time.
\end{theorem}

\emph{Overview.}
The natural idea is to reject a prefix of the arrivals, use their weights as
the prophet samples, and then run the prophet algorithm on the remaining
arrivals. The difficulty is that the prophet algorithm may later decide to
select one of the already rejected prefix elements, which cannot be
implemented. The reduction handles this in three steps.
Section~\ref{sec:reduction-filter} rejects a preliminary half-sample and
filters the remaining elements so that the expected total weight of all
candidates is at most the original optimum, while the expected optimum
of the candidates is at least half the original optimum. Section~\ref{sec:reduction-activation} presents each candidate to
the prophet algorithm as a random variable that is nonzero only with a small
activation probability $q$, independently as a sample and as a value; the
useful prophet reward is then linear in $q$, while an unimplementable
selection requires both activations and costs only $q^2$. This is where the
filter is needed: without it the quadratic loss would be charged against the
total weight rather than the optimum. Section~\ref{sec:reduction-simulation}
shows how to realize this simulation exactly online, by interleaving the
stored prefix with the physical arrivals so that the simulated arrival order is
independent of the activations. Optimizing $q$ yields the ratio
$\alpha^2/16$.

\subsection{A preliminary sample}\label{sec:reduction-filter}

Fix the secretary weights $w$. Reject a prefix of
length $K_0\sim\operatorname{Bin}(n,1/2)$, independently of the arrival
permutation. Its set $S$ is an independent half-sample of $E$. An element
$e\notin S$ is \emph{eligible} if $e\notin\cl_M(S_{\succ e})$, where
$S_{\succ e}$ consists of sample elements preceding $e$ in weight and
tie-priority order. Let $T$ be the eligible set. Eligibility is known
when $e$ arrives, since every sample weight has already been observed.

\begin{lemma}\label{lem:light-candidates}
The eligible set satisfies $\E[w(T)]=\E[w(B(S))]\le\OPT(w)$ and
$\E[\OPT(w|_T)]\ge\OPT(w)/2$.
\end{lemma}

\begin{proof}
The event $e\notin\cl_M(S_{\succ e})$ is independent of whether $e$
belongs to $S$. The two possibilities put $e$ in $B(S)$ and in $T$,
respectively, with equal probability. Summing their weights gives
the equality, and every $B(S)$ is independent, giving the upper bound.
For the other inequality, let $I=B(E)$. Each $e\in I$ is not spanned
even by all of its predecessors in the greedy order, so $I\setminus S$
is an independent subset of $T$. Its expected weight is $\OPT(w)/2$.
\end{proof}

\subsection{Rare activations and the loss bound}\label{sec:reduction-activation}

Conditional on $S$, let $E'=E\setminus S$ and define fixed values
$v_e=w_e\ind\{e\in T\}$ for $e\in E'$. The restriction $M|E'$ is known,
and each $v_e$ can be computed when $w_e$ is revealed. For a parameter
$q\in(0,1)$, we simulate the prophet algorithm on this restriction with
independent distributions
$\mathcal D_e=(1-q)\delta_0+q\delta_{v_e}$.
Write $X_e,Y_e$ for mutually independent Bernoulli-$q$ activation bits:
the sample is $X_ev_e$, and the actual value is $Y_ev_e$.

The simulation below observes every label with $X_e=1$ in a second
rejected prefix. It supplies all prophet samples before simulating
any actual arrivals, and processes the latter in a uniform
simulated arrival order independent of all activation bits. When the prophet algorithm
selects a positive-value element with $X_e=0$, that label is physically
arriving and we accept it immediately. Selections with $X_e=1$ cannot
be implemented physically and are discarded. All selections remain
in the internal state of the prophet algorithm, so the real output
is a subset of an independent set.

On every realization, the discarded value is at most $\sum_e X_eY_ev_e$. Thus, conditional
on $S$, the prophet guarantee gives
\begin{align}\label{eq:reduction-bound}
 \E[\text{secretary reward}\mid S]
 &\ge \alpha\E[\OPT((Y_ev_e)_{e\in E'})\mid S]-q^2w(T)\notag\\
 &\ge \alpha q\OPT(w|_T)-q^2w(T).
\end{align}
The second line follows by retaining the activated members of an
optimum of $T$. Taking expectations and applying
Lemma~\ref{lem:light-candidates} yields
$\E[\text{reward}]\ge(\alpha q/2-q^2)\OPT(w)$.
Choosing $q=\alpha/4$ gives $\alpha^2\OPT(w)/16$, as claimed.

We note that without the preliminary sample, the same calculation would charge
$q^2\sum_e w_e$, which can be arbitrarily larger than $\OPT(w)$.
The bound on $\E[w(T)]$ is needed for the analysis to go through.

\subsection{Exact online simulation}\label{sec:reduction-simulation}

We supply the implementation and its distributional justification.
Conditional on $S$, let $N=|E'|$; the residual physical permutation
$\pi$ is uniform. Draw a word $b\in\{0,1\}^N$ of independent
Bernoulli-$q$ bits, independently of $\pi$, and let $K$ be its number
of ones. Reject the next $K$ physical arrivals, denoting their set by
$R$, and store their weights. Set $X_e=\ind\{e\in R\}$. Initialize the
prophet algorithm with sample $v_e$ on $R$ and zero elsewhere, using
fresh independent internal randomness. Computing
$v_e$ for $e\in R$ uses only $S$ and this observed weight.

Draw the actual bits $Y_e$ independently. Scan the word $b$ from left
to right. At a one, take the next stored label of $R$ in its observed
order and simulate its actual arrival with value $Y_ev_e$. At a zero,
read the next physical arrival, compute its eligibility and $v_e$,
and immediately supply $Y_ev_e$ to the prophet algorithm. Implement
its positive-value acceptance in the physical process. The internal
state records every simulated acceptance, including discarded ones and
zero-value selections. Empty streams when $K=0$, $K=N$, or $N=0$
require no special changes.

Let $\Sigma$ be the simulated arrival permutation. For a fixed subset $D\subseteq E'$
and permutation $\sigma$ of $E'$, exactly one physical permutation and
one word $b$ produce $(R,\Sigma)=(D,\sigma)$: the physical permutation is
$\sigma|_D$ followed by $\sigma|_{E'\setminus D}$, and the word marks
the positions of $D$ in $\sigma$. Consequently,
\begin{equation}\label{eq:interleaving}
 \Pr(R=D,\Sigma=\sigma\mid S)
 =\frac{q^{|D|}(1-q)^{N-|D|}}{N!}.
\end{equation}
Thus $R$ is an independent Bernoulli-$q$ sample and $\Sigma$ is uniform
independently of $R$. The bits $Y$ are independent of both. The simulated
sample and actual vectors are therefore independent product draws from
the same distributions, in an independent random order, exactly as
required by the prophet guarantee. In particular, the simulation allows
$X_e=Y_e=1$; this is the overlap charged in~\eqref{eq:reduction-bound}.

After the preliminary sample, we test each arriving label's
eligibility with one independence query against the prefix of $B(S)$
preceding it in weight order, and make the accept/reject decision
immediately. The prophet algorithm runs on $M|E'$, and all remaining
bookkeeping is polynomial. Thus the reduction is online and preserves
polynomial running time.

For Algorithm~\ref{alg:prophet}, the resulting secretary algorithm is ordinal.
First suppose all secretary weights are positive and distinct. Every
simulated coordinate is either zero or an already observed secretary
weight, so the observed order determines every comparison used by the
prophet algorithm. Eligibility also uses only that order. The resulting
algorithm therefore has the same behavior for every positive weight vector
with a given strict order. Taking limits of such vectors extends its
$1/64$ guarantee to nonnegative weights, including ties resolved by the
fixed label priority. The algorithm treats every eligible real weight as
above the artificial zeros; it need not identify genuine zero weights.
Thus Theorem~\ref{thm:reduction} gives an ordinal $1/64$ guarantee when
instantiated with Algorithm~\ref{alg:prophet}.

\emph{A constant number of samples.}
The same argument applies if the prophet algorithm uses $k\ge1$ independent
samples per distribution. Set $p=1-(1-q)^k$, and use a Bernoulli-$p$
word to generate the second rejected prefix $R$ and the interleaving.
Conditional on $R$, draw the $k$ sample flags independently across labels:
on $R$, use the Bernoulli-$q$ product law conditioned on at least one
being one; outside $R$, set all flags to zero. Their unconditional law is
the required product law, independent of the simulated arrival order.
The expected discarded value is at most
$q[1-(1-q)^k]\E[w(T)]\le kq^2\OPT(w)$. Choosing $q=\alpha/(4k)$ gives
ratio $\alpha^2/(16k)$. This extension uses only the stated expected
prophet guarantee, not any conditional guarantee after fixing endpoint
pairs.

\section*{Acknowledgments and research process}
\addcontentsline{toc}{section}{Acknowledgments and research process}
Several of the results in this paper date from April 2026.
We first investigated the strong matroid secretary conjecture by computing
optimal ordinal policies for small matroids with the linear program of
Section~\ref{sec:lp}; the computed ratios always exceeded $1/e$. We also
tested several stronger conjectures on these instances, among them the
truncation conjecture of Section~\ref{sec:lp}, which we could
prove for uniform matroids and later refuted using our earlier high-girth
hardness result for graphic matroids. Our search for problems equivalent
to the strong conjecture led to the single-sample prophet inequality of
Section~\ref{sec:prophet}; the constant-to-constant implication for
secretary was known through personal communication, but as we did not have
access to that argument, Section~\ref{sec:reduction} gives an explicit
quantitative reduction.

Recent interactions with GPT Sol 5.6 and especially GPT Astra were
essential in obtaining the results above, in particular the linear-matroid
construction of Section~\ref{sec:linear} and the tight single-sample
prophet inequality of Section~\ref{sec:prophet}. Proofs produced in these
interactions were then leanified by GPT Astra and checked independently by
the authors.

\bibliographystyle{alphaurl}
\bibliography{references}

\appendix
\section{An exact linear program for ordinal matroid secretary}\label{app:lp}

Recall from Section~\ref{sec:lp} that $\rho(M)$ denotes the optimal ordinal
secretary ratio of a matroid $M$ of positive rank. We now give the full
linear program and prove that its optimum is $\rho(M)$, including the
attainment and policy-extraction statements in Theorem~\ref{thm:lp-exact}.
The computational motivation and comparison with earlier LP formulations
are in Section~\ref{sec:lp}.

We present the uncompressed program first, with a direct probabilistic
meaning for every variable. After proving correctness, we distinguish the
weighted and per-element objectives and describe symmetry and state
reductions that preserve the optimum.

\subsection{States and variables}\label{sec:lp-states}

Throughout, $M=(E,\mathcal I)$ has $n$ elements and rank function $r$.
Since policies are ordinal, only the relative order of the weights matters,
and with the fixed tie-breaking priority of Section~\ref{sec:preliminaries}
we may regard the weight vector as a strict total order on $E$; we call
this the \emph{weight order}. The elements arrive in a uniformly random
order, which is independent of the weight order.

An \emph{ordered subset} $\sigma$ is a list of distinct labels. We write
$S_\sigma$ for its underlying set and $|\sigma|=|S_\sigma|$ for its
length. For a weight order $\pi$ on $E$ and a subset $S\subseteq E$,
$\pi|_S$ denotes the ordered subset listing $S$ in decreasing weight
order. For $e\in S_\sigma$, $\sigma-e$ denotes the list obtained by
deleting $e$; thus $(\pi|_S)-e=\pi|_{S-e}$.

The observation available to an ordinal policy after $t$ arrivals is the
observed set $S$ together with the relative weight order on $S$, that is,
the ordered subset $\pi|_S$, together with the order in which those
elements arrived. The following observation is the basis of the program.

\begin{lemma}\label{lem:ordinal-invariance}
Fix an ordinal policy, a weight order $\pi$, and a $t$-subset $S$. The
probability that the first $t$ arrivals form $S$ and that the accepted set
after these arrivals equals $A$ depends on $\pi$ only through $\pi|_S$.
\end{lemma}

\begin{proof}
The first $t$ arrivals form $S$ with probability $\binom nt^{-1}$, and
conditional on this event the arrival order within $S$ is uniform on the
$t!$ orders of $S$. The policy's decisions during these arrivals are
determined by the arrival order within $S$, the relative weight order on
the observed elements, and the policy's internal coins. The relative weight
order on any subset of $S$ is determined by $\pi|_S$. Averaging over the
arrival order and the coins therefore yields a quantity depending on $\pi$
only through $\pi|_S$.
\end{proof}

We introduce one variable for each state and one for each transition.

\begin{itemize}
\item For every ordered subset $\sigma$, including the empty list, and
every independent $A\subseteq S_\sigma$, a variable $p_{\sigma,A}$. Its
intended value is the probability, under any weight order $\pi$ with
$\pi|_{S_\sigma}=\sigma$, that the first $|\sigma|$ arrivals form
$S_\sigma$ and the accepted set after them is $A$. By
Lemma~\ref{lem:ordinal-invariance} this is well defined. Note that
$p_{\sigma,A}$ includes the probability $\binom n{|\sigma|}^{-1}$ of the
observed set; it is not a conditional probability.
\item For every nonempty $\sigma$, every $e\in S_\sigma$, and every
independent $A\subseteq S_\sigma-e$, variables $y_{\sigma,e,A}$ and
$z_{\sigma,e,A}$. Their intended values are the probabilities, under any
weight order restricting to $\sigma$, that the first $|\sigma|-1$ arrivals
form $S_\sigma-e$ with accepted set $A$, that $e$ arrives next, and that
$e$ is accepted (for $y$) or rejected (for $z$).
\end{itemize}

Different orderings of the same set $S$ describe different weight orders,
and their variables are not combined: for each fixed $\pi$, only the
variables with $\sigma=\pi|_S$ for some $S$ are relevant. States with a
dependent accepted set are unreachable and are omitted.

For example, take $U_{2,3}$ on labels $a,b,c$, and suppose $a$ is heavier
than $b$. The state $\sigma=(a,b)$, $A=\{a\}$ can be reached in two ways:
$a$ arrived first and was accepted, then $b$ was rejected; or $b$ arrived
first and was rejected, then $a$ was accepted. Accordingly,
$p_{(a,b),\{a\}}=z_{(a,b),b,\{a\}}+y_{(a,b),a,\varnothing}$.
The total mass entering through the first of these histories, before the
decision on $b$, is
$y_{(a,b),b,\{a\}}+z_{(a,b),b,\{a\}}=p_{(a),\{a\}}/2$,
since $b$ is one of the two labels remaining after $a$.
These identities hold regardless of where $c$ lies in the weight order.

\subsection{The program}\label{sec:lp-program}

For a weight order $\pi$ and an element $e$, define the abbreviation
\begin{equation}\label{eq:lp-marginal}
 x^\pi_e=\sum_{S\subseteq E:\ e\in S}\;
 \sum_{\substack{A\subseteq S-e\\ A\in\mathcal I}}
 y_{\pi|_S,e,A},
\end{equation}
a linear function of the variables. Its intended value is the probability
that $e$ is accepted under weight order $\pi$: the summands correspond to
the disjoint possibilities for the observed set at the arrival of $e$ and
the accepted set at that moment. A \emph{prefix} of $\pi$ is a set
consisting of the $j$ heaviest elements under $\pi$, for some
$0\le j\le n$. The program is
\begin{subequations}\label{eq:lp}
\begin{alignat}{2}
 &\text{maximize}\quad c\quad\text{subject to}\notag\\
 & p_{\varnothing,\varnothing}=1,\label{eq:lp-root}\\
 & y_{\sigma,e,A}+z_{\sigma,e,A}
   =\frac{p_{\sigma-e,A}}{n-|\sigma|+1}
   &\quad&\text{for all nonempty $\sigma$, $e\in S_\sigma$, and
          independent $A\subseteq S_\sigma-e$,}\label{eq:lp-input}\\
 & y_{\sigma,e,A}=0
   &&\text{whenever $A+e$ is dependent,}\label{eq:lp-feasible}\\
 & p_{\sigma,A}=\sum_{e\in S_\sigma\setminus A}z_{\sigma,e,A}
   +\sum_{e\in A}y_{\sigma,e,A-e}
   &&\text{for all nonempty $\sigma$ and independent
          $A\subseteq S_\sigma$,}\label{eq:lp-output}\\
 & \sum_{e\in H}x^\pi_e\ge c\,r(H)
   &&\text{for every weight order $\pi$ and every prefix $H$ of $\pi$,}
   \label{eq:lp-competitive}\\
 & p,\,y,\,z\ge0.\notag
\end{alignat}
\end{subequations}
Constraint~\eqref{eq:lp-input} says that after reaching the state
$(\sigma-e,A)$, the label $e$ arrives next with probability
$1/(n-|\sigma|+1)$, since $n-|\sigma|+1$ labels remain unobserved and the
next arrival is uniform among them; the arriving label is then either
accepted or rejected. Constraint~\eqref{eq:lp-feasible} forbids dependent
acceptances. Constraint~\eqref{eq:lp-output} partitions the ways of
reaching the state $(\sigma,A)$ according to the last arrival $e$ and
whether it was rejected (so the accepted set was already $A$) or accepted
(so the accepted set was $A-e$). Constraint~\eqref{eq:lp-competitive} is
the prefix form of $c$-competitiveness, which by~\eqref{eq:layer-cake} is
equivalent to the weighted form for ordinal policies.

A consequence of the flow constraints is the mass identity
\begin{equation}\label{eq:lp-mass}
 \sum_{A}p_{\sigma,A}=\binom{n}{|\sigma|}^{-1}
 \qquad\text{for every ordered subset }\sigma,
\end{equation}
which is the probability that the first $|\sigma|$ arrivals form
$S_\sigma$. Indeed, summing~\eqref{eq:lp-output} over $A$ regroups the
right side as $\sum_{e\in S_\sigma}\sum_{A'}(y_{\sigma,e,A'}+z_{\sigma,e,A'})$,
with $A'$ ranging over independent subsets of $S_\sigma-e$;
by~\eqref{eq:lp-input} and induction on $|\sigma|=t$ this equals
$t\binom n{t-1}^{-1}/(n-t+1)=\binom nt^{-1}$.

\subsection{Correctness}\label{sec:lp-correctness}

\begin{theorem}\label{thm:lp-exact}
Let $M$ be a matroid with at least one nonloop. The linear
program~\eqref{eq:lp} has an optimal solution, and its optimal value is
$\rho(M)$. Moreover, every feasible solution with objective value $c$
yields a $c$-competitive ordinal policy whose decisions depend only on the
current observed weight order, the arriving label, and the current accepted
set. In particular,
the supremum defining $\rho(M)$ is attained by such a policy.
\end{theorem}

\begin{proof}
We show that every $c$-competitive ordinal policy yields a feasible
solution with objective value $c$, and that every feasible solution with
objective value $c$ yields a $c$-competitive ordinal policy. Together these
give that the value of the program is $\rho(M)$, once we know that the
maximum is attained.

\emph{From policies to solutions.}
Fix an ordinal policy that is $c$-competitive. It may be randomized and may
remember the entire arrival history, not only the current weight order and
accepted set. Define $p_{\sigma,A}$, $y_{\sigma,e,A}$, and $z_{\sigma,e,A}$
as the probabilities described in Section~\ref{sec:lp-states}. By
Lemma~\ref{lem:ordinal-invariance}, applied to $S_\sigma$ and to
$S_\sigma-e$, these are well defined: the transition probabilities are
determined by the same averaging, since the event that $e$ arrives at
position $|\sigma|$ and the subsequent decision are functions of the
arrival order within $S_\sigma$, the relative weight order on $S_\sigma$,
and the coins. The variables are nonnegative, and \eqref{eq:lp-root}
holds because nothing has been observed or accepted initially.
Conditional on any history through the first $|\sigma|-1$ arrivals, the
next label is uniform among the $n-|\sigma|+1$ unobserved labels, which
gives~\eqref{eq:lp-input}; the policy never accepts a dependent addition,
which gives~\eqref{eq:lp-feasible}; and partitioning the event of reaching
$(\sigma,A)$ according to the last arrival gives~\eqref{eq:lp-output}.
Under a weight order $\pi$, the probability that $e$ is accepted is the sum
of the probabilities of the disjoint events indexed by the observed set at
its arrival and the accepted set at that moment, which is $x^\pi_e$
in~\eqref{eq:lp-marginal}. Finally, $c$-competitiveness of an ordinal
policy is equivalent to the prefix inequalities
$\sum_{e\in H}\Pr(e\in A)\ge c\,r(H)$ for all weight orders and prefixes,
by~\eqref{eq:layer-cake} and the limiting argument following it. This
is~\eqref{eq:lp-competitive}.

\emph{From solutions to policies.}
Fix a feasible solution with objective value $c$. Define a policy as
follows. When a label $e$ arrives, let $\sigma$ be the observed weight
order including $e$, and let $A$ be the set accepted before this arrival.
Accept $e$ with probability
$y_{\sigma,e,A}/(y_{\sigma,e,A}+z_{\sigma,e,A})$, interpreted as zero when
the denominator is zero. This uses only the current weight order and the
accepted set together with the arriving label, and it is ordinal. By~\eqref{eq:lp-feasible}, $e$ is
accepted with positive probability only if $A+e$ is independent, so the
output is always independent.

We claim that under this policy, for every weight order $\pi$ and every
$S\subseteq E$, the probability of reaching the state $(\pi|_S,A)$ is
$p_{\pi|_S,A}$, and the probabilities of the accepting and rejecting
transitions into states with observed set $S$ are the corresponding
$y$ and $z$ variables. We argue by induction on $|S|$. For $S=\varnothing$
this is~\eqref{eq:lp-root}. Let $|S|=t\ge1$, put $\sigma=\pi|_S$, and fix
$e\in S$ and independent $A\subseteq S-e$. By induction, the probability of
reaching $(\sigma-e,A)$ is $p_{\sigma-e,A}$. Conditional on this event,
$e$ is the next arrival with probability $1/(n-t+1)$, so the probability
of reaching $(\sigma-e,A)$ and then seeing $e$ arrive is
$y_{\sigma,e,A}+z_{\sigma,e,A}$ by~\eqref{eq:lp-input}. Multiplying by
the acceptance probability shows that the accepting and rejecting
transitions have probabilities exactly $y_{\sigma,e,A}$ and
$z_{\sigma,e,A}$; when both variables vanish, so do both probabilities.
Summing the transitions into $(\sigma,A)$ over the last arrival
gives~\eqref{eq:lp-output}, so the probability of reaching $(\sigma,A)$
is $p_{\sigma,A}$. This proves the claim.

By the claim, $x^\pi_e$ is the probability that $e$ is accepted under
$\pi$, so~\eqref{eq:lp-competitive} states that
$\sum_{e\in H}\Pr(e\in A)\ge c\,r(H)$ for every prefix $H$ of every weight
order. By~\eqref{eq:layer-cake}, the policy is $c$-competitive.

\emph{Attainment.}
The always-rejecting policy yields a feasible solution with $c=0$, so the
program is feasible. If $e$ is a nonloop and $\pi$ places $e$ first, the
prefix $H=\{e\}$ has $r(H)=1$, and the claim above shows that $x^\pi_e$ is a
probability; hence $c\le1$ for every feasible solution. A feasible linear
program with bounded objective attains its optimum. The optimal solution
yields an optimal ordinal policy of the stated form.
\end{proof}

\subsection{Two objectives}\label{sec:lp-objectives}

Program~\eqref{eq:lp} optimizes the full weighted competitive ratio.
Our linear-matroid construction in Section~\ref{sec:linear} satisfies a
stronger, probability-competitive requirement: every element of the greedy
basis is selected with probability at least $c$. This requirement also has
a linear programming formulation, obtained by
replacing~\eqref{eq:lp-competitive} with the constraints $x^\pi_e\ge c$
for every weight order $\pi$ and every element $e$ of its greedy basis,
namely every $e$ whose addition increases the rank of the prefix of $\pi$
ending at $e$. Every feasible solution of the stronger program is feasible
for~\eqref{eq:lp}, but the two optima can differ.

For example, on the uniform matroid $U_{2,3}$ the full and per-element
optima are $3/4$ and $2/3$, respectively. To see this, symmetrize over
labels, and assume without loss that the policy accepts all remaining
elements whenever capacity permits doing so. Conditional on accepting the
first arrival, let $u,v$ be its time-two acceptance probabilities for a
record and a nonrecord. Enumerating the six weight orders gives
best-to-worst marginals $((4+u-v)/6,\,2/3,\,(4-u+v)/6)$. Conditional on
rejecting the first arrival, all three marginals are $2/3$. Thus the
expected number of top-two selections is at most $3/2$, attained at
$u=1,v=0$ after accepting the first arrival; that policy satisfies all
prefix constraints with ratio $3/4$. The second-best marginal is always
$2/3$, so the per-element optimum is $2/3$, attained by selecting two
fixed arrival positions.

\subsection{Symmetry and state reductions}\label{sec:lp-reductions}

The uncompressed program has a variable for every ordered subset and
independent accepted set, and the number of ordered subsets alone exceeds
$n!$. Three reductions make the program tractable for larger instances
without changing its value; they are what we used in our computations.

First, only rank-increasing prefixes need constraints
in~\eqref{eq:lp-competitive}: if a prefix $H+e$ has the same rank as $H$,
then nonnegativity of $x^\pi_e$ shows that the inequality for $H$ implies
the inequality for $H+e$.

Second, averaging any policy over the automorphisms of $M$ preserves its
worst-case guarantee, since an automorphism maps weight orders to weight
orders and preserves ranks of prefixes. Hence there is an optimal policy
that is invariant under automorphisms, and ordered subsets and weight
orders may be reduced to automorphism representatives, transporting labels
and accepted sets by the same permutation.

Third, for a fixed observed set $S$, accepted sets $A,A'$ can be merged
when $A\cup J$ and $A'\cup J$ are either both independent or both
dependent for every $J\subseteq E\setminus S$. Equivalently, the labeled
minors $(M/A)|(E\setminus S)$ and $(M/A')|(E\setminus S)$ coincide, so the
two states admit exactly the same future decisions. For example, after two
arrivals in $U_{2,3}$, the empty accepted set and an accepted singleton
both leave a free singleton on the last label. This equivalence is
preserved after either accepting or rejecting the next label: rejection
deletes that label from the remaining minor, and acceptance contracts it.
Past selections are already recorded in the acceptance flows, so averaging
the next decision over histories within a merged state preserves the
aggregate future flows and the objective. The conditional averaging depends
only on the observed order and the merged state, so it is common to every
weight order extending that observation.

These reductions decrease the program size without changing its value.
They do not produce a polynomial-size formulation, nor does the existence
of the program itself supply a uniform competitive bound over all matroids;
it is a tool for computing and testing rather than for proving the strong
conjecture.

\section{Ordinal secretary ratios under truncation}\label{app:truncation}

This appendix proves the two results about truncation stated in
Section~\ref{sec:lp}. Recall that $\rho(M)$ denotes the optimal ordinal
secretary ratio of a matroid $M$ of positive rank, that $M^{(k)}$ denotes
its rank-$k$ truncation, and that Conjecture~\ref{conj:truncation} asserts
that $\rho(M^{(k)})$ is nondecreasing in $k$.
Section~\ref{sec:truncation-uniform} proves the conjecture for uniform
matroids in the strict form of Theorem~\ref{thm:uniform-monotonicity}, and
Section~\ref{sec:truncation-graphic} refutes it in general with an explicit
graphic matroid (Theorem~\ref{thm:truncation-counterexample}), following
the method of the high-girth hardness theorem of Banihashem et
al.~\cite{graphic2025}. Since the upper bound in the counterexample allows
the algorithm to observe numerical weights, the conjecture also fails
without the ordinal restriction.

\subsection{Uniform matroids}\label{sec:truncation-uniform}

Let $r(n,k)$ denote the maximum expected fraction of the top $k$
elements selected in a uniform random arrival order by a policy that
observes only relative ranks and has capacity $k$.
Kleinberg's ordinal algorithm~\cite{kleinberg2005}
gives $r(n,k)\ge1-O(k^{-1/2})$. Chan, Chen, and Jiang~\cite{chan2015}
characterized asymptotically optimal thresholds for the more general
$(J,K)$-secretary problem, with separate selection and relevance parameters.
We first relate the top-$k$ objective
to the full weighted ordinal secretary problem, then compare the exact
optima at successive capacities. The comparison expresses $k\,r(n,k)$ as
the value of an optimal stopping problem with $k$ slots and independent
rewards, namely the posterior probabilities that each arrival is among the
top $k$ given its relative rank. Passing from $k$ to $k-1$ both lowers these
rewards and removes a slot; Lemmas~\ref{lem:capacity-thinning}
and~\ref{lem:adjacent-posteriors} bound the two effects separately.

The following equivalence between the top-$k$ objective and the weighted
ordinal ratio is standard: Buchbinder, Jain, and
Singh~\cite{buchbinder2014} observe that a monotone policy
for the $(k,k)$-secretary problem transfers its ratio to the weighted
problem, and that an optimal policy can be taken monotone. We include a
self-contained proof; the strict finite-$n$ monotonicity in
Theorem~\ref{thm:uniform-monotonicity} is the new claim.

\begin{lemma}\label{lem:uniform-ratio}
For the uniform matroid $U_{k,n}$, $\rho(U_{k,n})=r(n,k)$.
\end{lemma}

\begin{proof}
Symmetrizing over label permutations does not decrease a uniform
matroid policy's worst-case guarantee, so labels can be ignored.
Let $R_t$ be the relative rank of arrival $t$, with rank one best.
The variables $R_t$ are independent and uniform on $[t]$, since
relative-rank sequences are in bijection with permutations. The set
of global ranks observed by time $t$ is a uniform $t$-subset,
independent of its relative arrival order. Conditional
on the entire ordinal history through $t$, the posterior probability
that the arrival is in the top $k$ is
\begin{equation}\label{eq:posterior}
 p_k(t,R_t),\qquad
 p_k(t,s)=\Pr\{H_{k,t}\ge s\},\quad
 H_{k,t}\sim\operatorname{Hypergeometric}(n,k,t).
\end{equation}
Thus an optimal policy can use a Bellman rule depending only on time,
remaining capacity, and the current posterior. Since $p_k(t,s)$ is
nonincreasing in $s$, the rule is nonincreasing in the current relative
rank at each fixed preceding history.

For any label-blind policy with this property, the probabilities of
selecting global ranks $1,2,\ldots,n$ are nonincreasing. Pair a
permutation with the permutation exchanging ranks $j$ and $j+1$,
using the same random seed. Their histories and decisions agree
until the later arrival of the pair. The earlier arrival receives
the same decision, and the later arrival is at least as likely to
be selected when it is rank $j$. Averaging proves the assertion.

Apply this to the optimal top-$k$ policy. Every prefix of size
$m\le k$ has average capture probability at least that of the top
$k$. Every larger prefix has expected capture at least the top-$k$
score and has rank $k$. Hence all prefix-rank guarantees are at
least $r(n,k)$, implying $\rho(U_{k,n})\ge r(n,k)$. The reverse
inequality follows by taking weights arbitrarily close to one on
the top $k$ and arbitrarily close to zero outside, preserving the
strict order. Any ordinal weighted guarantee must therefore also
hold for the top-$k$ objective.
\end{proof}

Albers and Ladewig~\cite{albersladewig2021} observed, in numerical
computations for small $k$, that the competitive ratio of a particular
threshold algorithm increases with $k$. The following theorem compares the
optimal ratios themselves, for every finite $n$.

\begin{theorem}\label{thm:uniform-monotonicity}
For $2\le k\le n$, $r(n,k)>r(n,k-1)$. In particular, the optimal
ordinal secretary ratio is strictly increasing through successive
uniform matroid ranks.
\end{theorem}

\emph{Plan of the proof.}
By the posterior argument in the proof of Lemma~\ref{lem:uniform-ratio},
the top-$k$ problem is an optimal stopping problem with $k$ slots: at time
$t$ the policy sees the relative rank $R_t$, and accepting the arrival earns
expected reward $p_k(t,R_t)$, its posterior probability of being among the
top $k$. For independent nonnegative finite-support rewards
$X=(X_1,\ldots,X_n)$, let $V_b(X)$ be the optimal expected total reward
from at most $b$ sequential selections, and define $X_t^k=p_k(t,R_t)$.
Then
\begin{equation}\label{eq:uniform-value}
 k\,r(n,k)=V_k(X^k),
\end{equation}
since the expected number of top-$k$ elements captured is the expected sum
of the posteriors of the accepted arrivals. The rewards $X^k_t$ are
independent because the relative ranks $R_t$ are independent, and past
ranks carry no information about future rewards beyond what the rewards
themselves encode.

Passing from $k$ to $k-1$ changes two things: the reward sequence changes
from $X^k$ to $X^{k-1}$, and one slot is lost. We bound the two effects
separately. Lemma~\ref{lem:adjacent-posteriors} shows that the rewards
$X^{k-1}$ are dominated, in the sense relevant for optimal stopping, by the
rewards obtained from $X^k$ by independently retaining each with probability
$(k-1)/k$ and zeroing the rest. Lemma~\ref{lem:capacity-thinning} shows that
with such thinned rewards, $k-1$ slots earn at most a $(k-1)/k$ fraction of
what $k$ slots earn on the unthinned rewards. Combining the two gives
$(k-1)\,r(n,k-1)\le\frac{k-1}{k}\,k\,r(n,k)$, that is,
$r(n,k-1)\le r(n,k)$, and a small explicit slack in the second lemma makes
the inequality strict.

\begin{lemma}[Thinning and capacity]\label{lem:capacity-thinning}
Let $1\le m<K$, and let $B_t$ be independent Bernoulli-$(m/K)$
bits independent of $X$. Then
$V_m((B_tX_t)_t)\le(m/K)V_K(X)$. The inequality is strict if
$V_K(X)>V_m(X)$.
\end{lemma}

\begin{proof}
The idea is to build a $K$-slot policy on $X$ out of many $m$-slot
policies on thinned copies of $X$, and to compare values. Color each time
independently and uniformly from $[K]$; then for each $m$-element color set
$J$, the rewards $X_t\ind\{C_t\in J\}$ are exactly an $(m/K)$-thinning of
$X$. For each such $J$, run the same deterministic optimal
$m$-selection policy on $X_t\ind\{C_t\in J\}$. Choose this policy
to reject zero rewards, and let $I_t^J$ be its selections. Put
$a_t=\binom{K-1}{m-1}^{-1}\sum_J I_t^J$. Only color sets containing
$C_t$ can select at time $t$, so $0\le a_t\le1$. Each copy uses at
most $m$ slots, giving $\sum_t a_t\le K$ on every realization.

Round these decisions online with one independent uniform
$U\in[0,1)$. If $Q_t=\sum_{i\le t}a_i$, accept at time $t$ when
$\lfloor Q_t+U\rfloor-\lfloor Q_{t-1}+U\rfloor=1$. This uses at
most $K$ slots. Conditional on all rewards and colors, its marginal
at time $t$ is $a_t$, so its expected reward is
$(K/m)V_m((B_tX_t)_t)$. Optimality proves the weak inequality.

With probability $\delta=K^{1-n}$, all colors coincide. On this
event the active copies see the same full reward sequence and make
identical decisions; the fractional rule is integral and uses at
most $m$ slots. For every fixed color sequence the rounded rule is
a feasible online policy on $X$, since the colors are independent
of $X$. Therefore the construction gives the stronger bound
\begin{equation}\label{eq:strict-thinning-paper}
 \frac Km V_m((B_tX_t)_t)
 \le(1-\delta)V_K(X)+\delta V_m(X),
\end{equation}
which proves strictness under the stated condition.
\end{proof}

\begin{lemma}[Adjacent posteriors]\label{lem:adjacent-posteriors}
For every convex function $\phi$ with $\phi(0)=0$,
$\E\phi(X_t^{k-1})\le\frac{k-1}{k}\E\phi(X_t^k)$.
\end{lemma}

\begin{proof}
We couple the two posteriors by a single change to the hypergeometric
population: change a uniformly chosen one of the $k$ successes to a
failure. This turns a population with $k$ successes into one with $k-1$,
so the sample count for $k-1$ is obtained from the sample count for $k$
by possibly removing one success. A sample with exactly $s$
successes drops below $s$ with probability $s/k$; no other count
changes the tail event. Thus, suppressing $t$,
$p_{k-1}(s)=\frac{k-s}{k}p_k(s)+\frac{s}{k}p_k(s+1)$ for
$1\le s\le k-1$. Jensen's inequality and summation give
\[
 \sum_{s=1}^{k-1}\phi(p_{k-1}(s))
 \le\frac1k\sum_{s=1}^{k-1}
       \bigl[(k-s)\phi(p_k(s))+s\phi(p_k(s+1))\bigr]
 =\frac{k-1}{k}\sum_{s=1}^k\phi(p_k(s)).
\]
Terms beyond $t$ vanish. Dividing by $t$ proves the claim.
\end{proof}

\begin{proof}[Proof of Theorem~\ref{thm:uniform-monotonicity}]
If independent rewards $Z,W$ satisfy
$\E(Z_t-z)_+\le\E(W_t-z)_+$ for every $z\ge0$, then
$V_b(Z)\le V_b(W)$. The Bellman update
$\E\max\{v(b),Z_t+v(b-1)\}=v(b)+\E(Z_t-[v(b)-v(b-1)])_+$
is increasing in both continuation values and has a nonnegative threshold,
so backward induction applies.

Put $m=k-1$, and thin $X^k$ with independent Bernoulli-$(m/k)$
bits. Lemma~\ref{lem:adjacent-posteriors}, with
$\phi(x)=(x-z)_+$, gives the preceding comparison between $X^m$
and $(B_tX_t^k)_t$. Using~\eqref{eq:uniform-value} and
Lemma~\ref{lem:capacity-thinning},
\[
 m r(n,m)=V_m(X^m)
 \le V_m((B_tX_t^k)_t)
 \le\frac mk V_k(X^k)=m r(n,k).
\]
For strictness, the first $k$ rewards of $X^k$ are all at least
$\eta=\binom nk^{-1}>0$: indeed
$p_k(t,s)\ge p_k(t,t)=\binom{n-t}{k-t}/\binom nk\ge\eta$ for
$s\le t\le k$. Simulate an optimal $m$-selection policy and use
one extra slot on its first rejection among these $k$ times. This
improves reward by at least $\eta$ on every realization, so
$V_k(X^k)-V_m(X^k)\ge\eta$. Bound~\eqref{eq:strict-thinning-paper}
makes the displayed comparison strict. It also yields
$r(n,k)-r(n,k-1)\ge[k^n\binom nk]^{-1}$.
\end{proof}

\begin{corollary}\label{cor:uniform-strong}
For every $n\ge2$ and $1\le k\le n$, $\rho(U_{k,n})\ge\rho(U_{1,n})\ge1/e$.
Thus the strong matroid secretary conjecture holds for every uniform
matroid.
\end{corollary}

\begin{proof}
By Lemma~\ref{lem:uniform-ratio} and Theorem~\ref{thm:uniform-monotonicity},
$\rho(U_{k,n})=r(n,k)\ge r(n,1)=\rho(U_{1,n})$. The rank-one matroid
$U_{1,n}$ is the classical secretary problem on $n$ elements, for which the
ordinal cutoff rule with cutoff $s$ selects the heaviest element with
probability $c_n(s)$, and Section~\ref{sec:linear-cutoff} exhibits a cutoff
with $c_n(s)\ge1/e$ for every $n\ge2$.
\end{proof}

The conclusion of Corollary~\ref{cor:uniform-strong} is not new: the
virtual algorithm of Babaioff, Immorlica, Kempe, and
Kleinberg~\cite[Theorem~1]{babaioff2007knapsack} is ordinal and selects
every element of the top $k$ with the classical cutoff probability, which
gives $\rho(U_{k,n})\ge1/e$ directly. The corollary shows that the same
conclusion also follows from strict monotonicity alone, without any
algorithm beyond the rank-one cutoff rule.

\subsection{A graphic counterexample}\label{sec:truncation-graphic}

As explained at the start of this appendix, nonmonotonicity follows from
the high-girth hardness theorem of Banihashem et al.~\cite{graphic2025}
together with the uniform three-choice bound below, since a graph of girth
at least four has uniform rank-three truncation. Here we give a
self-contained example on $K_{2,N}$, following their method of embedding
independent single-secretary instances.

\begin{theorem}\label{thm:truncation-counterexample}
Let $N=10^{42}$ and let $M$ be the cycle matroid of $K_{2,N}$. Then
\[
 \rho(M^{(3)})>\frac{53}{100}>\rho(M).
\]
Since $M$ has rank $N+1$, some $4\le k\le N+1$ satisfies
$\rho(M^{(k)})<\rho(M^{(k-1)})$.
\end{theorem}

There are no cycles of length at most three in this graph, so
$M^{(3)}=U_{3,2N}$. We first give a uniform lower bound above $0.53$
for this three-choice problem, and then bound the unrestricted
graphic problem below that value.

\emph{A three-choice policy.}
On $h\ge6$ objects, reject through time $\lfloor h/6\rfloor$, accept
records until time $\lfloor h/2\rfloor$, and thereafter accept
arrivals among the best three observed so far. Stop after three
acceptances. Call an arrival \emph{eligible} if the policy would accept it
were a slot available: a record in the first active phase, and an arrival
among the best three so far in the second. Ignoring the capacity, the
eligibility indicators are independent, with probabilities $1/t$ and $3/t$ in the two active
phases. At scaled time $x=t/h$, the mean number of earlier eligible
arrivals converges to $\lambda(x)=\log(6x)$ in the first phase and
$\lambda(x)=\log3+3\log(2x)$ in the second. The sum of squared
eligibility probabilities is $O(1/h)$, so the probability of
fewer than three earlier eligible arrivals tends to
$e^{-\lambda}(1+\lambda+\lambda^2/2)$.

In the first phase, the expected fraction gained before imposing capacity
is $1/(3t)$ times the posterior that the current record is in the global
top three. This posterior tends to $1-(1-x)^3$. In the second phase
every top-three arrival is eligible, so the expected fraction gained is
exactly $1/h$ per arrival. The posterior identity~\eqref{eq:posterior}
and independence from past ranks let us multiply by the probability
that a slot remains. The~resulting Riemann sum has limit
\begin{align*}
 L_*={}&\int_{1/6}^{1/2}
 \frac{[1+\log(6x)+\tfrac12\log^2(6x)](3-3x+x^2)}{18x}\,dx\\
 &+\int_{1/2}^{1}
 \frac{1+\lambda(x)+\tfrac12\lambda(x)^2}{24x^3}\,dx.
\end{align*}
The Poisson approximation and record posterior converge uniformly within
each phase, and the $O(1/h)$ boundary contributions vanish.
Writing $a=\log2$ and $b=\log3$, integration gives
\[
 L_*=-\frac{3a^2}{32}-\frac{ab}{16}-\frac{5a}{32}
       +\frac{b^3}{36}+\frac{11b^2}{144}+\frac{47b}{144}
       +\frac{425}{1728}.
\]
Substituting $0.693<a<0.694$ and $1.098<b<1.099$, using the
upper endpoints in the negative terms and lower endpoints in the
positive terms, gives
$L_*>531/1000>53/100$. These logarithm enclosures
follow, for example, from six terms of
$\log((1+z)/(1-z))=2\sum_{j\ge0}z^{2j+1}/(2j+1)$ at
$z=1/3,1/2$, bounding the positive tail geometrically.

This limit lower-bounds $r(n,3)$ for every finite $n\ge3$. To
simulate a policy on any larger horizon $h$, insert $h-n$ fictitious
elements, all worse than the real elements, at uniformly random
positions and in a random order. All comparisons are computable
online. Discard any fictitious selections while counting them in
the simulated policy's capacity. The top-three score is unchanged,
and the combined order is uniform. Taking $h\to\infty$ gives
$r(n,3)\ge L_*$. Lemma~\ref{lem:uniform-ratio} proves the required
lower bound for $M^{(3)}$.

\emph{A two-element distribution.}
Fix integers $Q>1$ and $L\ge1$. Choose an index
$i\in\{0,\ldots,L-1\}$ with probability $Q^{-(i+1)}/Z$, where
$Z=\sum_{i=0}^{L-1}Q^{-(i+1)}$, and present the two values
$Q^i,Q^{i+1}$ in random order. Their expected maximum is
$\mu=L/Z$. An algorithm may know this distribution and observe
numerical values. When the first arrival is revealed, all the algorithm
sees is its value $Q^j$, which is consistent with two scales; let $p_j$ be
the probability of taking the first arrival when it equals $Q^j$. Awarding
the second value after a rejection can only help, so the expected reward at
scale $i$ is at most
\[
 \tfrac12\bigl[p_iQ^i+(1-p_i)Q^{i+1}\bigr]
 +\tfrac12\bigl[p_{i+1}Q^{i+1}+(1-p_{i+1})Q^i\bigr]
 =Q^{i+1}\Bigl[\frac{1+1/Q}{2}+\frac{(1-1/Q)(p_{i+1}-p_i)}{2}\Bigr],
\]
the two terms corresponding to which value arrives first.
Each scale's probability times $Q^{i+1}$ is $1/Z$, so summing over $i$ the
differences $p_{i+1}-p_i$ telescope; using $p_L-p_0\le1$ bounds the
expected reward by $\alpha\mu$, where
$\alpha=(1+1/Q)/2+(1-1/Q)/(2L)$. The bound also permits any
independent auxiliary randomness, absorbed into the probabilities
$p_j$.

\emph{Embedding the pairs.}
The $N$ degree-two vertices of $K_{2,N}$ partition the edges into
pairs. Independently draw the preceding distribution for each pair
and randomly assign its values to its two labels. Reveal all edges
in uniform random order. For each pair, retain only the first edge
selected by an arbitrary graphic secretary algorithm. Fix the other
pairs' values, the arrival permutation, and the algorithm's random seed;
all are independent of this pair's scale and value assignment. Its
first-arrival decision is then a function only of the value revealed
there, so the two-element bound applies conditionally. Averaging and
summing over pairs bounds the expected total weight of these first
selections by $N\alpha\mu$. This holds even if the algorithm knows
the graph and the entire input distribution.

A forest can contain both edges of at most one pair, since two
complete pairs form a four-cycle. There is thus at most one extra
selected edge, of weight at most $Q^L$. On the other hand, the
heavier edge from each pair forms a forest, so $\E\OPT\ge N\mu$.
Thus $\E[\text{reward}]/\E\OPT\le\alpha+Q^L/(N\mu)$.
For $Q=100$, $L=21$, and $N=Q^L$, using $Z<1/99$ makes this
strictly less than $37/70+1/2079<53/100$.
Any guarantee valid on every fixed instance also holds after averaging
over this distribution, so $\rho(M)<53/100$. Together with the
uniform lower bound, this proves Theorem~\ref{thm:truncation-counterexample}.

The same matroid also disproves monotonicity for the optimal secretary
ratio when algorithms can use numerical weights. The uniform lower
bound uses only ordinal information, and the graphic upper bound permits
numerical observations.

\end{document}